\documentclass[12pt]{amsart}
\usepackage{amssymb}
\usepackage{color}
\usepackage{amsmath,epic,curves,amscd}
\usepackage[english]{babel}
\usepackage{graphicx}
\usepackage{comment}
\usepackage{appendix}
\usepackage[all]{xy}
\newtheorem{claim}{}[section]
\newtheorem{theorem}[claim]{Theorem}
\newtheorem{definition}[claim]{Definition}

\newtheorem{lemma}[claim]{Lemma}
\newtheorem{proposition}[claim]{Proposition}
\newtheorem{corollary}[claim]{Corollary}

\theoremstyle{remark}

\renewenvironment{proof}{\noindent{\it Proof. \hskip0pt}}
                      {$\square$\par\medskip}

\begin{document}
\baselineskip 6.0 truemm
\parindent 1.5 true pc

\newcommand\lan{\langle}
\newcommand\ran{\rangle}
\newcommand\tr{{\text{\rm Tr}}\,}
\newcommand\ot{\otimes}
\newcommand\ol{\overline}
\newcommand\join{\vee}
\newcommand\meet{\wedge}
\renewcommand\ker{{\text{\rm Ker}}\,}
\newcommand\image{{\text{\rm Im}}\,}
\newcommand\id{{\text{\rm id}}}
\newcommand\tp{{\text{\rm tp}}}
\newcommand\pr{\prime}
\newcommand\e{\epsilon}
\newcommand\la{\lambda}
\newcommand\inte{{\text{\rm int}}\,}
\newcommand\ttt{{\text{\rm t}}}
\newcommand\spa{{\text{\rm span}}\,}
\newcommand\conv{{\text{\rm conv}}\,}
\newcommand\rank{\ {\text{\rm rank of}}\ }
\newcommand\re{{\text{\rm Re}}\,}
\newcommand\ppt{\mathbb T}
\newcommand\rk{{\text{\rm rank}}\,}
\newcommand\SN{{\text{\rm SN}}\,}
\newcommand\SR{{\text{\rm SR}}\,}
\newcommand\HA{{\mathcal H}_A}
\newcommand\HB{{\mathcal H}_B}
\newcommand\HC{{\mathcal H}_C}
\newcommand{\bra}[1]{\langle{#1}|}
\newcommand{\ket}[1]{|{#1}\rangle}
\newcommand\cl{\mathcal}
\newcommand\idd{{\text{\rm id}}}
\newcommand\OMAX{{\text{\rm OMAX}}}
\newcommand\OMIN{{\text{\rm OMIN}}}
\newcommand\diag{{\text{\rm Diag}}\,}

\title{Various notions of positivity for bi-linear maps and applications to tri-partite entanglement}

\author{Kyung Hoon Han and Seung-Hyeok Kye}
\address{Department of Mathematics, The University of Suwon, Gyeonggi-do 445-743, Korea}
\email{kyunghoon.han at gmail.com}
\address{Department of Mathematics and Institute of Mathematics, Seoul National University, Seoul 151-742, Korea}
\email{kye at snu.ac.kr}
\thanks{KHH and SHK were partially supported by NRF-2012R1A1A1012190 and NRFK grant 2013-004942, respectively.}



\begin{abstract}
We consider bi-linear analogues of $s$-positivity for linear maps.
The dual objects of these notions can be described in terms of
Schimdt ranks for tri-tensor products and Schmidt numbers for
tri-partite quantum states. These tri-partite versions of Schmidt
numbers cover various kinds of bi-separability, and so we may
interpret witnesses for those in terms of bi-linear maps. We give
concrete examples of witnesses for various kinds of three qubit
entanglement.
\end{abstract}

\maketitle

%%%%%%%%%%%%%%%%%%%%%%%%%%%%%%%%%%%%%%%%%%%%%%%%%%%%%%%%%%%%%%%%%%%%%%%%%%%%%%%%%%%%%%%%%%%%%%%%%%%%%%%%%%%%%%%%%%
%%%%%%%%%%%%%%%%%%%%%%%%%%%%%%%%%%%%%%%%%%%%%%%%%%%%%%%%%%%%%%%%%%%%%%%%%%%%%%%%%%%%%%%%%%%%%%%%%%%%%%%%%%%%%%%%%%
%%%%%%%%%%%%%%%%%%%%%%%%%%%%%%%%%%%%%%%%%%%%%%%%%%%%%%%%%%%%%%%%%%%%%%%%%%%%%%%%%%%%%%%%%%%%%%%%%%%%%%%%%%%%%%%%%%
%%%%%%%%%%%%%%%%%%%%%%%%%%%%%%%%%%%%%%%%%%%%%%%%%%%%%%%%%%%%%%%%%%%%%%%%%%%%%%%%%%%%%%%%%%%%%%%%%%%%%%%%%%%%%%%%%%
%%%%%%%%%%%%%%%%%%%%%%%%%%%%%%%%%%%%%%%%%%%%%%%%%%%%%%%%%%%%%%%%%%%%%%%%%%%%%%%%%%%%%%%%%%%%%%%%%%%%%%%%%%%%%%%%%%
\section{Introduction}

Order structures are key ingredients in various subjects of
mathematics as well as functional analysis, where linear maps
preserving positivity play important roles. We call those positive
linear maps. After representation theorem by Stinespring
\cite{stine}, complete positivity has been considered as a right
morphism to study operator algebras, which are non-commutative in
general. By definition of complete positivity, it is clear that
there exist hierarchy structures between positivity and complete
positivity, and this leads to define $s$-positivity of linear maps
for natural numbers $s=1,2,\dots$. The notion of $s$-positivity
turns out to be very useful in itself. For examples, various
inequalties like Schwartz and Kadison inequalities in operator
algebras already hold for $2$-positive linear maps \cite{choi-schw}.
Importance of $s$-positivity is also recognized in current quantum
information theory. Considering the dual objects of  $s$-positive
linear maps in matrix algebras, we get natural classification of
bi-partite entanglement in terms of Schmidt ranks. See
\cite{eom-kye,sbl,TH}. Furthermore, distillability problem which is one of the most important
in quantum information theory can be formulated \cite{DSSTT} in terms of Schmidt ranks.

The purpose of this note is to introduce the
bi-linear analogues of $s$-positive linear maps, and classify
tri-partite entanglement as dual objects. Our classification scheme
include various kinds of bi-separability.

Positive linear maps and bi-partite entanglement are related through duality
\cite{eom-kye,horo-1}, which goes back to the work by Woronowicz
\cite{woronowicz} in the seventies. Very recently, the second author
\cite{kye_3_qubit} has shown that $n$-partite genuine separability
is dual to positivity of $(n-1)$-linear maps whose linearization
gives rise to positivity with respect to the function system maximal
tensor product  \cite{effros,han} of matrix algebras. This means that it is
necessary and sufficient to construct a positive multi-linear map,
in order to detect entanglement which is not genuinely separable. A
natural question arises: What kinds of positivity are suitable to
detect another kinds of entanglement, for example, genuine
entanglement which is not in the convex hull of bi-separable states
with respect to all possible bi-partitions.

In the tri-partite cases, there are three kinds of bi-separability:
$A$-$BC$, $B$-$CA$ and $C$-$AB$ separabilities according to
bi-partitions. For a bi-linear map $\phi: M_A\times M_B\to M_C$
between matrix algebras to be dual objects of those notions, we need
the following properties:
\begin{enumerate}
\item[(A)]
For each $x\in M_A^+$, the map $y\mapsto \phi(x,y)$ is completely positive.
\item[(B)]
For each $y\in M_B^+$, the map $x\mapsto \phi(x,y)$ is completely positive.
\item[(C)]
The linearization $M_A\otimes M_B\to M_C$ is positive with respect to
the usual order.
\end{enumerate}
In order to detect tri-partite genuine entanglement, we need a bi-linear map
satisfying the above three conditions simultaneously.
In this paper, we formulate the notions of positivity for bi-linear maps
which explain the above three properties in a single framework.

It was shown in \cite{KPTT} that the linearization of
a bi-linear map $\phi:{\mathcal S}\times{\mathcal T}\to{\mathcal R}$
between operator systems is completely positive with respect to the operator system maximal
tensor products if and only if the following condition
\begin{equation}\label{st-pos}
[x_{i,j}]\in M_p({\mathcal S})^+,\
[y_{k,\ell}]\in M_q({\mathcal T})^+\
\Longrightarrow\
[\phi(x_{i,j},y_{k,\ell})]\in M_{pq}({\mathcal R})^+
\end{equation}
holds for every $p,q=1,2,\dots$. It is tempting to call the property (\ref{st-pos}) as
$(p,q)$-positivity. Then $\phi$ satisfies the condition (A) if and only if
it is $(1,\infty)$-positive, that is, $(1,q)$-positive for every $q=1,2,\dots$.
Condition (B) is, of course, nothing but $(\infty,1)$-positivity.
But, it is not possible to formulate the condition (C)
with $(p,q)$-positivity. In fact, there is no way
to control the matrix size over the range spaces with the above
condition (\ref{st-pos}). This motivates the following definition.

\begin{definition}
\emph{Let $S=(s_1,s_2,\dots,s_n)$ be an $n$-tuple of natural numbers.
An $(n-1)$-linear map $\phi:{\mathcal S}_1\times\cdots\times {\mathcal S}_{n-1}
\to{\mathcal S}_n$ between operator systems is said to be}
$S$-positive
\emph{if the following condition holds:}
$$
\begin{aligned}
x_k = [x^k_{i_k,j_k}] \in M_{s_k}({\mathcal S}_k)^+\ {\text{\rm for}}\ k=1,2,\dots,&n-1\
{\text{\rm and}}\
\alpha\in M_{s_n,s_1\cdots s_{n-1}}\ \\
&\Longrightarrow\
\alpha[\phi(x^1_{i_1,j_1},\dots,x^{n-1}_{i_{n-1},j_{n-1}})]\alpha^*\in M_{s_n}(\mathcal S_n)^+
\end{aligned}
$$
\end{definition}
Then the above conditions (A), (B) and (C) become
$(1,\infty,\infty)$, $(\infty,1,\infty)$
and $(\infty,\infty,1)$-positivity, respectively.
See Proposition \ref{1pp}.
In the case of linear maps  with $n=2$,
it is $(p,q)$-positive if and only if it is $(p\meet q,p\meet q)$-positive
if and only if it is $p\meet q$-positive in the usual sense, where
$p\meet q$ denotes the minimum of $p$ and $q$. Therefore,
the above definition reduces to the usual notion of $s$-positivity
in the case of linear maps. Furthermore, a bi-linear map satisfies the condition
(\ref{st-pos}) if and only if it is $(p,q,pq)$-positive.

If we restrict ourselves in the cases when domains and ranges are
matrix algebras, then we can consider the Choi matrices of bi-linear
maps. We find conditions when they are positive, that is, positive
semi-definite. In the course of discussion, we get bi-linear version
of the isomorphism between completely positive linear maps and
positive block matrices, as well as decomposition of completely
positive maps into the sum of elementary operators
\cite{choi75-10,kraus}. The linearization of
$(p,q,r)$-positive bi-linear maps will be described in terms of
suitable quantizations of domains and ranges.

We introduce the notion that Schmidt
numbers for tri-partite states $\varrho$ are less than or equal to triplets $(p,q,r)$ of natural numbers.
We write this property by $\SN(\varrho)\le(p,q,r)$.
To do this, we first define the Schmidt ranks for vectors in the tensor products
of three vector spaces. For this purpose, we use the natural isomorphisms
between tensor products and linear mapping spaces, and consider the dimensions
of supports and ranges of the corresponding maps.
We establish the duality between $(p,q,r)$-positive
bi-linear maps and states $\varrho$ with the property $\SN(\varrho)\le (p,q,r)$.

After we summarize briefly in the next section several notions in operator systems we need,
we present in Section 3 properties of $(p,q,r)$-positive bi-linear maps mentioned above.
We give the definition of $\SN(\varrho)\le(p,q,r)$
in Section 4, and prove the duality in Section 5, where we also interpret the notion $(p,q,r)$-positivity
in various ways. We exhibit in Section 6
concrete examples of bi-linear maps with various kinds of positivity
in two dimensional matrix algebras.
They will be witnesses for various kinds of three qubit entanglement,
including witnesses for genuine entanglement. We close this paper to discuss several
related problems in the last section.

Throughout this note, we will use notations ${\mathcal H}_A, {\mathcal H}_B$ and ${\mathcal H}_C$
for complex Hilbert spaces $\mathbb C^a, \mathbb C^b$ and $\mathbb C^c$, respectively.
Matrix algebras on them will be also denoted by $M_A, M_B$ and $M_C$, and so
they are $a\times a, b\times b$ and $c\times c$ matrix algebras, respectively.
For a bi-linear map $\phi:{\mathcal S}\times{\mathcal T}\to{\mathcal R}$,
we denote by $\tilde\phi:{\mathcal S}\otimes{\mathcal T}\to{\mathcal R}$
its linearization which sends $x\otimes y$ to $\phi(x,y)$.
For $x=[x_{i,j}]\in M_p({\mathcal S})$ and
$y=[y_{k,\ell}]\in M_q({\mathcal T})$, we write
$\phi_{p,q}(x,y)=[\phi(x_{i,j},y_{k,\ell})]\in M_{pq}({\cl R})$
for notational convenience. If $\phi:{\mathcal S}\to{\mathcal R}$ is a linear map
then we also write $\phi_p(x)=[\phi(x_{i,j})]\in M_p({\cl R})$. Recall that $\phi$ is $p$-positive
if $x\in M_p({\mathcal S})^+$ implies $\phi_p(x)\in M_p({\mathcal R})^+$.

The authors are grateful to Kil-Chan Ha and Jaeseong Heo for
fruitful discussion on the topics. This work has been completed when
the second author was visiting Singapore. He is also grateful to
Denny Leung and Wai Shing Tang for their warm hospitality and
stimulating discussion during his stay.

%%%%%%%%%%%%%%%%%%%%%%%%%%%%%%%%%%%%%%%%%%%%%%%%%%%%%%%%%%%%%%%%%%%%%%%%%%%%%%%%%%%%%%%%%%%%%%%%%%%%%%%%%%%%%%%%%%
%%%%%%%%%%%%%%%%%%%%%%%%%%%%%%%%%%%%%%%%%%%%%%%%%%%%%%%%%%%%%%%%%%%%%%%%%%%%%%%%%%%%%%%%%%%%%%%%%%%%%%%%%%%%%%%%%%
%%%%%%%%%%%%%%%%%%%%%%%%%%%%%%%%%%%%%%%%%%%%%%%%%%%%%%%%%%%%%%%%%%%%%%%%%%%%%%%%%%%%%%%%%%%%%%%%%%%%%%%%%%%%%%%%%%
%%%%%%%%%%%%%%%%%%%%%%%%%%%%%%%%%%%%%%%%%%%%%%%%%%%%%%%%%%%%%%%%%%%%%%%%%%%%%%%%%%%%%%%%%%%%%%%%%%%%%%%%%%%%%%%%%%
%%%%%%%%%%%%%%%%%%%%%%%%%%%%%%%%%%%%%%%%%%%%%%%%%%%%%%%%%%%%%%%%%%%%%%%%%%%%%%%%%%%%%%%%%%%%%%%%%%%%%%%%%%%%%%%%%%
\section{Tensor product and quantization of operator systems}

A unital self-adjoint space of bounded operators on a Hilbert space is said to be an operator system.
If ${\mathcal S}$ is an operator system acting on a Hilbert space ${\cl H}$, then the space $M_n({\mathcal S})$
of all $n\times n$ matrices over ${\mathcal S}$ acts on the Hilbert space ${\mathcal H}^n$. The order structures
of $M_n({\mathcal S})$ for each $n \in \mathbb N$ are related by the following compatibility relation:
\begin{equation}\label{MO}
x\in M_n({\cl S})^+,\ \alpha\in M_{m,n}\ \Longrightarrow \alpha x\alpha^*\in M_m({\mathcal S})^+.
\end{equation}
The identity operator ${\rm id}_{\mathcal H^n}$ on the Hilbert space ${\cl H}^n$ plays the role of order unit
of $M_n({\cl S})$, that is, for every self-adjoint $x\in M_n(\cl S)$ there is $r>0$ such that
$x\le r \cdot {\rm id}_{\mathcal H^n}$. This order unit also has the Archimedean property:
If $x\in M_n({\cl S})$ and $\varepsilon \cdot {\rm id}_{\mathcal H^n}+x\ge 0$ for each $\varepsilon>0$ then $x\in M_n({\cl S})^+$.

If $V$ is a $*$-vector space and $C_n$ is a cone in $M_n(V)_h$
satisfying (\ref{MO}), then we call $\{ C_n\}_{n=1}^\infty$ a matrix
ordering and $(V, \{ C_n\}_{n=1}^\infty$) a matrix ordered
$*$-vector space. Choi and Effros \cite{CE} showed that matrix
ordered $*$-vector spaces equipped with Archimedean matrix order
units can be realized as unital self-adjoint spaces of bounded
operators on a Hilbert space. Thus, we also call them operator
systems. An operator system structure of ${\cl S}$ determines an
operator space structure on $\cl S$. Especially, if $x$ is a
Hermitian element of $M_n({\cl S})$ then we have
$$
\|x\|_n=\inf \{ r>0 : -r I_n \otimes 1_{\mathcal S} \le x\le r I_n \otimes 1_{\mathcal S} \}.
$$

We proceed to recall the definition \cite{KPTT} of the maximal tensor product of operator systems.
For two operator systems $\cl S$ and $\cl T$, the sets
$$
D_n^{\max}(\cl S, \cl T) = \{ \alpha(P \otimes Q)
\alpha^* : P \in M_k(\cl S)^+, Q \in M_\ell(\cl T)^+, \alpha \in
M_{n,k\ell},\ k,\ell \in \mathbb N \}
$$
for each $n=1,2,\dots$ give rise to a matrix ordering on
$\cl S \otimes \cl T$ with a matrix order unit $1_{\cl S} \otimes 1_{\cl
T}$. Let $\{ M_n(\cl S \otimes_{\max} \cl T)^+ \}_{n=1}^\infty$ be
the Archimedeanization of the matrix ordering $\{ D_n^{\max}(\cl
S, \cl T) \}_{n=1}^\infty$. Then it can be written as
\begin{equation}\label{max}
M_n(\cl S \otimes_{\max} \cl T)^+ = \{ z \in M_n(\cl S
\otimes \cl T) : \forall \varepsilon>0, z+\varepsilon I_n \otimes 1_{\cl S} \otimes 1_{\cl
T} \in
D_n^{\max}(\cl S, \cl T) \}.
\end{equation}
We call the operator system $(\cl S \otimes \cl T, \{ M_n(\cl S
\otimes_{\max} \cl T)^+ \}_{n=1}^\infty, 1_{\cl S} \otimes 1_{\cl
T})$ the maximal tensor product of $\cl S$ and $\cl T$, and denote
by $\cl S \otimes_{\max} \cl T$. The family $\{ M_n(\cl S
\otimes_{\max} \cl T)^+ \}_{n=1}^\infty$ is the smallest among
positive cones of operator system structures on $\cl S \otimes \cl
T$. For unital $C^*$-algebras $\cl A$ and $\cl B$, we have the
completely order isomorphic inclusion $\cl A \otimes_{\max} \cl B
\subset \cl A \otimes_{\rm C^*\max} \cl B$.
In particular, $M_m \otimes_{\max} M_n \simeq M_{mn}$ is a complete order isomorphism.
On the other hand, one can also define the largest positive cones on ${\cl S}\ot{\cl T}$
to get the minimal tensor product ${\cl S}\otimes_{\rm min}{\cl T}$.

For given Archimedean order unit spaces $V$, there are two canonical ways
to endow matrix order structures with which they are operator systems \cite{PTT}.
These processes are usually called quantization.
One way is to endow
the largest positive cones of operator system structures on $V$
whose first level positive cone coincides with $V^+$, to get
the minimal operator system $\OMIN(V)$. The other is to endow the smallest positive cones,
to get the the maximal operator system $\OMAX(V)$.

For a given operator system ${\cl S}$ and a natural number $k \in \mathbb N$, one can also define
two operator systems, super $k$-maximal operator systems $\OMAX^k(\mathcal S)$ and
super $k$-minimal operator systems $\OMIN^k(\mathcal S)$, respectively \cite{X}.
For each
$n,k \in \mathbb N$, we set
$$
D_n^{\max, k} = \left\{ \alpha \begin{pmatrix} s_1 && \\ &\ddots& \\
&& s_m \end{pmatrix} \alpha^* : \alpha \in M_{n, mk}, s_i \in M_k(\mathcal S)^+, m \in \mathbb N \right\}.
$$
Applying Archimedeanization process, we get
\begin{equation}\label{max-k}
C_n^{\max, k} = \{ s \in M_n(\mathcal S) : \forall \varepsilon>0, x+\varepsilon I_n \otimes 1_{\mathcal S} \in D_n^{\max, k}(\mathcal S) \}.
\end{equation}
Then, $(\mathcal S, \{ C^{\max}_k (\mathcal S) \}_{k=1}^\infty,
1_{\mathcal S})$ is an operator system which is denoted by
$\OMAX^k(\mathcal S)$. In particular, we have $\OMAX^1(\mathcal S) =
\OMAX (\mathcal S)$. The family $\{ C^{\max}_k (\mathcal S)
\}_{k=1}^\infty$ is the smallest among positive cones of operator
system structures on $\mathcal S$ whose $k$-th level positive cones
coincide with $M_k(\mathcal S)^+$. A linear map $\phi : \OMAX^k(\cl
S) \to \cl T$ is $k$-positive if and only if it is completely
positive. Moreover, this property characterizes $\OMAX^k (\cl S)$
\cite[Theorem 4.6]{X}.

For each
$n, k \in \mathbb N$, we set
$$
C_n^{\min, k} = \{ [s_{i,j}]\in M_n (\mathcal S) : [\varphi (s_{i,j})] \ge 0\ {\text{\rm for each}}\ \varphi \in S_k(\mathcal S), \},
$$
where $S_k (\mathcal S)$ denotes the set of unital completely positive linear maps from $\mathcal S$ into $M_k$.
Then, $(\mathcal S, \{ C^{\min}_k (\mathcal S) \}_{k=1}^\infty, 1_{\mathcal S})$ is
an operator system which is denoted by ${\rm OMIN}^k(\mathcal S)$. In particular,
we have ${\rm OMIN}^1(\mathcal S) = {\rm OMIN} (\mathcal S)$.
The family $\{ C^{\min}_k (\mathcal S) \}_{k=1}^\infty$ is the largest among positive cones of
operator system structures on $\mathcal S$
whose $k$-th level positive cones coincide with $M_k(\mathcal S)^+$.
 A linear map $\phi : \cl S \to \OMIN^k(\cl T)$ is $k$-positive if and only if it is completely positive. Moreover,
 this property characterizes $\OMIN^k (\cl T)$ \cite[Theorem 3.7]{X}.

Duals of operator systems are matrix ordered by the cones
$$
M_n(\mathcal S^*)^+ = {\rm CP}(\mathcal S, M_n), \qquad n=1,2,\dots,
$$
where ${\rm CP}(\mathcal S, M_n)$ denotes the set of all completely positive linear maps from $\mathcal S$ into $M_n$.
With this matrix ordering, we have the complete order isomorphism \cite[Lemma 5.7, Theorem 5.8]{KPTT}
\begin{equation}\label{duality-max}
(\mathcal S \otimes_{\max} \mathcal T)^* \simeq \mathcal L (\mathcal S, \mathcal T^*),
\end{equation}
where $\mathcal L (\mathcal S, \mathcal T^*)$ is matrix ordered by
$$
M_n(\mathcal L (\mathcal S, \mathcal T^*))^+ = {\text{\rm CP}}(\mathcal S, M_n(\mathcal T^*)).
$$
Unfortunately, duals of operator systems fail to be operator systems in general
due to the lack of matrix order unit. However, duals of matrix algebras are again operator systems
because the trace satisfies the condition of Archimedean matrix order unit. Moreover,
matrix algebras are self-dual operator systems:
Every $x\in M_n$ corresponds to $f_x\in M_n^*$ given by
$f_x(y) = \tr (xy^\ttt)=\sum_{i,j=1}^n x_{i,j}y_{i,j}$. The map
$$
\gamma : x \in M_n \mapsto f_x \in M_n^* %\qquad (f_x(y) = \tr (xy^\ttt)=\sum_{i,j=1}^n x_{i,j}y_{i,j}).
$$
is a unital complete order isomorphism \cite[Theorem 6.2]{PTT}. Related with quantization,
$\gamma$ gives rise to the duality \cite[Proposition 6.5]{X}:
\begin{equation}\label{duality-super}
{\rm OMAX}^k (M_n) \simeq {\rm OMIN}^k(M_n)^*, \qquad {\rm OMIN}^k (M_n) \simeq {\rm OMAX}^k(M_n)^*.
\end{equation}
We will use $\gamma$ to define the dual map of a bi-linear map from $M_A\times M_B$ into $M_C$
which is given by a permutation on $\{A,B,C\}$.

%%%%%%%%%%%%%%%%%%%%%%%%%%%%%%%%%%%%%%%%%%%%%%%%%%%%%%%%%%%%%%%%%%%%%%%%%%%%%%%%%%%%%%%%%%%%%%%%%%%%%%%%%%%%%%%%%%
%%%%%%%%%%%%%%%%%%%%%%%%%%%%%%%%%%%%%%%%%%%%%%%%%%%%%%%%%%%%%%%%%%%%%%%%%%%%%%%%%%%%%%%%%%%%%%%%%%%%%%%%%%%%%%%%%%
%%%%%%%%%%%%%%%%%%%%%%%%%%%%%%%%%%%%%%%%%%%%%%%%%%%%%%%%%%%%%%%%%%%%%%%%%%%%%%%%%%%%%%%%%%%%%%%%%%%%%%%%%%%%%%%%%%
%%%%%%%%%%%%%%%%%%%%%%%%%%%%%%%%%%%%%%%%%%%%%%%%%%%%%%%%%%%%%%%%%%%%%%%%%%%%%%%%%%%%%%%%%%%%%%%%%%%%%%%%%%%%%%%%%%
%%%%%%%%%%%%%%%%%%%%%%%%%%%%%%%%%%%%%%%%%%%%%%%%%%%%%%%%%%%%%%%%%%%%%%%%%%%%%%%%%%%%%%%%%%%%%%%%%%%%%%%%%%%%%%%%%%
\section{$S$-positive bi-linear maps}

Following proposition shows that some combinations of numbers in the definition of
$(p,q,r)$-positivity are redundant.

\begin{proposition}\label{redundant1}
Suppose that $\phi:{\mathcal S}\times{\mathcal T}\to{\mathcal
R}$ is a bi-linear map in operator systems
$\mathcal S, \mathcal T$ and $\mathcal R$. For $p,q =1,2,\dots$, the following are equivalent:
\begin{enumerate}
\item[(i)] $\phi$ satisfies the condition {\rm (\ref{st-pos})};
\item[(ii)] $\phi$ is $(p,q,r)$-positive for each $r =1,2,\dots$;
\item[(iii)] $\phi$ is $(p,q,r)$-positive for some $r \ge pq$;
\item[(iv)] $\phi$ is $(p,q,pq)$-positive.
\end{enumerate}
\end{proposition}

\begin{proof}
The implication (i) $\Longrightarrow$ (ii) follows from the relation (\ref{MO}), and
(ii) $\Longrightarrow$ (iii) is clear. For the direction (iii) $\Longrightarrow$ (iv),
we note that
$$
\begin{pmatrix}
\alpha \phi_{p,q}(x,y) \alpha^* &0 \\ 0&0_{r-pq}\end{pmatrix}
= \begin{pmatrix} \alpha \\ 0_{r-pq,pq} \end{pmatrix} \phi_{p,q}(x,y) \begin{pmatrix} \alpha^* & 0_{pq,r-pq} \end{pmatrix}
\in M_r(\mathcal R)^+
$$
for $x \in M_p(\mathcal S)^+, y \in M_q(\mathcal T)^+$ and $\alpha \in M_{pq}$. This implies that
$\alpha \phi_{p,q}(x,y) \alpha^*\in M_{pq}({\cl S})^+$, as it was required.
Finally, we take $\alpha = I_{pq}$ for (iv) $\Longrightarrow$ (i).
\end{proof}

\begin{proposition}\label{redundant2}
Suppose that $\phi:{\mathcal S}\times{\mathcal T}\to{\mathcal
R}$ is a bi-linear map in operator systems $\mathcal S, \mathcal T$ and $\mathcal R$. We have the following:
\begin{enumerate}
\item[(i)] $\phi$ is $(1,q,r)$-positive if and only if $\phi$ is $(1,q \wedge r, q \wedge r)$-positive.
\item[(ii)] $\phi$ is $(p,1,r)$-positive if and only if $\phi$ is $(p \wedge r,1,p \wedge r)$-positive.
\end{enumerate}
\end{proposition}

\begin{proof}
Since we may exchange the role of ${\cl S}$ and ${\cl T}$, it suffices to prove (i).
This is immediate when $q \le r$ by Proposition \ref{redundant1}. Let $q \ge r$ and $x \in \mathcal S^+$.
If $\phi$ is $(1,q,r)$-positive then we have
$$
\alpha \phi_{1,r}(x,y) \alpha^* = \begin{pmatrix} \alpha & 0_{r,q-r} \end{pmatrix}
\phi_{1,q}(x,y \oplus 0_{q-r}) \begin{pmatrix} \alpha^* \\ 0_{q-r,r} \end{pmatrix}
\in M_r(\mathcal R)^+,$$
for $y \in M_r(\mathcal T)^+$ and $\alpha \in M_{r}$, and so $\phi$ is $(1,r,r)$-positive.
For the converse, suppose that $\phi$ is $(1,r,r)$-positive. Then we have
$$
\alpha \phi_{1,q}(x,y) \alpha^* = \phi_{1,r}(x, \alpha y \alpha^*) \in M_r(\mathcal R)^+
$$
for $y \in M_q(\mathcal T)^+$ and $\alpha \in M_{r,q}$. This shows that $\phi$ is $(1,q,r)$-positive.
\end{proof}

Taking $\mathcal S = \mathbb C$ in Proposition \ref{redundant2} (i),
we see that a linear map is $(q,r)$-positive if and only if it is $(q\meet r,q\meet r)$-positive
if and only if it is $q\meet r$-positive in the usual sense. When $\mathcal S, \mathcal T$ and  $\mathcal R$ are matrix algebras,
we will see later that the role of $p,q$ and $r$ in Propositions \ref{redundant1} and \ref{redundant2}
may be permuted together with $\mathcal S, \mathcal T, \mathcal R$.
See Corollary \ref{dual-per}.
If one of $p,q$ is $1$ then we may assume that
the others coincide by Proposition \ref{redundant2}. These are the most important cases with which
conditions (A), (B) and (C) discussed in Introduction may be
explained.

\begin{proposition}\label{1pp}
Suppose that $\phi:{\mathcal S}\times{\mathcal T}\to{\mathcal
R}$ is a bi-linear map in operator systems $\mathcal S, \mathcal T$ and $\mathcal R$. We have the following:
\begin{enumerate}
\item[(i)]
$\phi$ is $(1,p,p)$-positive if and only if the map $y\mapsto\phi(x,y)$ is $p$-positive for each
$x\in {\mathcal S}^+$.
\item[(ii)]
$\phi$ is $(p,1,p)$-positive if and only if the map $x\mapsto\phi(x,y)$ is $p$-positive for each
$y\in {\mathcal T}^+$.
\item[(iii)]
$\phi$ is $(p,p,1)$-positive if and only if $\sum_{i,j=1}^p\phi(x_{ij},y_{ij})\in {\mathcal R}^+$
for each $x \in M_p({\mathcal S})^+$ and $y \in M_p({\mathcal T})^+$.
\item[(iv)]
When $\mathcal S = M_p$, $\phi$ is $(p,p,1)$-positive if and only if $\tilde{\phi} : M_p(\mathcal T) \to \mathcal R$ is positive.
\end{enumerate}
\end{proposition}

\begin{proof}
Statements (i) and (ii) follow immediately from Proposition \ref{redundant1}.

(iii). We denote by $\{ e_i \}_{i=1}^p$ the canonical basis of $\mathbb C^p$ written as column vectors.
Then the identity
$$
\sum_{i,j=1}^p \phi (x_{i,j}, y_{i,j})
= \begin{pmatrix} e_1^\ttt & \cdots & e_p^\ttt \end{pmatrix}
\phi_{p,p}(x, y)
\begin{pmatrix} e_1 \\ \vdots \\ e_p \end{pmatrix}
$$
shows that if $\phi$ is $(p,p,1)$-positive then $\sum_{i,j=1}^p \phi (x_{i,j}, y_{i,j})\in \mathcal R^+$
whenever $x \in M_p({\mathcal S})^+$ and $y \in M_p({\mathcal T})^+$.

For the other direction, let
$x\in M_p({\cl S})^+, y\in M_p({\cl T})^+$ and $\alpha \in M_{1, p^2}$.
If we denote by $\tilde \alpha$ the  $p \times p$ matrix whose entries are given by $\tilde \alpha_{ij} = \alpha_{1,(i-1)p + j}$, then we have
$$
\alpha=\begin{pmatrix} e_1^\ttt & \cdots & e_p^\ttt \end{pmatrix} \begin{pmatrix} \tilde{\alpha} && \\ & \ddots & \\ && \tilde{\alpha}\end{pmatrix}
=\begin{pmatrix} e_1^\ttt & \cdots & e_p^\ttt \end{pmatrix}(I_p\otimes \tilde\alpha).
$$
Therefore, we have
$$
\begin{aligned}
\alpha \phi_{p,p} (x,y) \alpha^*
= & \begin{pmatrix} e_1^\ttt & \cdots & e_p^\ttt \end{pmatrix}(I_p\otimes \tilde\alpha)
\phi_{p,p}(x, y)
(I_p\otimes \tilde\alpha)^* \begin{pmatrix} e_1^\ttt & \cdots & e_p^\ttt \end{pmatrix}^*\\
= & \begin{pmatrix} e_1^\ttt & \cdots & e_p^\ttt \end{pmatrix}
(I_p \otimes \tilde{\alpha}) \tilde{\phi}_{p^2}(x \otimes y) (I_p \otimes \tilde{\alpha})^*
\begin{pmatrix} e_1^\ttt & \cdots & e_p^\ttt \end{pmatrix}^* \\
= & \begin{pmatrix} e_1^\ttt & \cdots & e_p^\ttt \end{pmatrix}
\phi_{p,p}(x, \tilde{\alpha} y \tilde{\alpha}^*)
\begin{pmatrix} e_1^\ttt & \cdots & e_p^\ttt \end{pmatrix}^*.
\end{aligned}
$$
If we write $z=\tilde{\alpha} y \tilde{\alpha}^* \in M_n(\mathcal T)^+$ then
this is nothing but $\sum_{i,j=1}^p \phi(x_{i,j},z_{i,j}) \in \mathcal R^+$ by assumption,
as it was required.

(iv). Suppose that $\phi: M_p\times {\mathcal T}\to{\mathcal R}$ is $(p,p,1)$-positive, and
$y \in M_p({\mathcal T})^+$. Since $[e_{i,j}]_{i,j} \in M_p(M_p)^+$, we have
$$
\tilde \phi (y) = \tilde \phi \left(\sum_{i,j=1}^p e_{i,j} \otimes y_{i,j}\right)
= \sum_{i,j=1}^p \tilde \phi (e_{i,j}, y_{i,j}) \in \mathcal R^+
$$
by (iii).
For the converse, suppose that $\tilde\phi:M_p({\mathcal T})\to{\mathcal R}$ is positive. We note that
$$
\alpha(x \otimes y) \alpha^* \in (M_p \otimes_{\max} \mathcal T)^+ = M_p(\mathcal T)^+,
$$
for $x \in M_p(M_p)^+$, $y\in M_p({\mathcal T})^+$ and $\alpha\in M_{1,p^2}$. It follows that
$$
\alpha \phi_{p,p}(x,y) \alpha^* = \alpha \tilde \phi_{p^2} (x \otimes y)\alpha^*
= \tilde \phi ( \alpha(x \otimes y) \alpha^*) \in \mathcal R^+,
$$
which shows that $\phi$ is $(p,p,1)$-positive.
\end{proof}

Now, we consider bi-linear maps $\phi:M_A\times M_B\to M_C$
between matrix algebras $M_A$, $M_B$ and $M_C$. In this case, a bi-linear map
may be described in terms of associated Choi matrix,
as it was defined in \cite{kye_3_qubit} for multi-linear cases.
For a given bi-linear map $\phi:M_A\times M_B\to M_C$, the Choi matrix $C_\phi$ is defined by
$$
C_\phi
=\sum_{i,j=1}^a \sum_{k,\ell=1}^b
 |i\ran\lan j|\otimes |k\ran\lan \ell|\ot
 \phi(|i\ran\lan j|,|k\ran\lan \ell|)\in M_A\ot M_B\ot M_C.
$$
For a given matrix $C\in M_A\otimes M_B\otimes M_C$, we may write
$$
\begin{aligned}
C
&=\sum_{i,j=1}^a |i\ran\lan j|\ot C_{i,j}\in M_A\ot (M_B\otimes M_C)\\
&=\sum_{i,j=1}^a \sum_{k,\ell=1}^b |i\ran\lan j|\ot |k\ran\lan \ell|\ot C_{(i,k),(j,\ell)}\in M_A\ot M_B\otimes M_C.
\end{aligned}
$$
We associate the bi-linear map $\phi_C:M_A\otimes M_B\to M_C$ by
$$
\phi_C
(|i\ran\lan j|, |k\ran\lan \ell|)=C_{(i,k),(j,\ell)}\in M_C.
$$
The correspondences $\phi\mapsto C_\phi$ and $C\mapsto \phi_C$ are just the
Choi-Jamio\l kowski isomorphisms \cite{choi75-10,jami}
when $M_B=\mathbb C$.

We consider the elementary bi-linear map $\phi_V:M_A\times M_B\to M_C$ with an $c \times ab$ matrix $V$, defined by
\begin{equation}\label{elementary}
\phi_V(x,y)=V(x\otimes y)V^*, \qquad x\in M_A, y\in M_B.
\end{equation}
It is obvious that the map $\phi_V$ satisfies the condition (\ref{st-pos}) for every $p,q=1,2,\dots$, and so it is
$(p,q,r)$-positive for every $p,q,r=1,2,\dots$ by Proposition \ref{redundant1}. To calculate its Choi matrix, we write
$$
|V_{(i,k)}\ran = V | i\ran |k\ran  \in {\mathcal H}_C,
$$
which is the $(i,k)$-th column of the $c \times ab$ matrix $V$. Then we see that
$$
\begin{aligned}
C_{\phi_V}
&=\sum_{i,j=1}^a \sum_{k,\ell=1}^b  |i\ran\lan j|\otimes |k\ran\lan \ell|\ot
 V(|i\ran\lan j|\otimes|k\ran\lan \ell|)V^*\in M_A\ot M_B\ot M_C\\
&=\sum_{i,j=1}^a \sum_{k,\ell=1}^b |i\ran |k\ran \lan j| \lan\ell | \otimes |V_{(i,k)}\ran\lan V_{(j,\ell)}|
   \in (M_A\ot M_B)\ot M_C\\
&=\left(\sum_{(i,k)=(1,1)}^{(a,b)} |i\ran |k\ran |V_{(i,k)}\ran\right)
   \left(\sum_{(j,\ell)=(1,1)}^{(a,b)} \lan j| \lan \ell  |\lan V_{(j,\ell)}|\right)\in M_A\ot M_B\ot M_C
\end{aligned}
$$
is a positive matrix of rank one whose range vector is given by
$\sum_{(i,k)=(1,1)}^{(a,b)} |i\ran |k\ran |V_{(i,k)}\ran$. Conversely, If $C_\phi\in M_A\ot M_B\ot M_C$
is positive with rank one then $\phi$ is of the form in (\ref{elementary}),
where $V$ is given by the above relation in the obvious way. This actually proves the equivalence between
statements (v) and (vi) in the following:

\begin{theorem}\label{Choi-iso}
For a bi-linear map $\phi:M_A\times M_B\to M_C$, the following are equivalent:
\begin{enumerate}
\item[(i)]
$\phi$ is $(p,q,r)$-positive for each $p,q,r=1,2,\dots$;
\item[(ii)]
$\phi$ is $(a,b,ab)$-positive;
\item[(iii)]
$\phi$ satisfies the condition {\rm (\ref{st-pos})} for each $p,q=1,2,\dots$;
\item[(iv)]
$\phi$ satisfies the condition {\rm (\ref{st-pos})} with $p=a$ and $q=b$;
\item[(v)]
the Choi matrix $C_\phi$ is positive;
\item [(vi)]
$\phi=\sum \phi_{V_i}$ with $c \times ab$ matrices $V_i$'s.
\end{enumerate}
\end{theorem}

\begin{proof}
Equivalences (i) $\Longleftrightarrow$ (iii) and (ii) $\Longleftrightarrow$ (iv) come from Proposition \ref{redundant1}.
We proceed to prove the implications (iv) $\Longrightarrow$ (v) $\Longrightarrow$ (vi) $\Longrightarrow$ (iii).
The condition (iv) tells us that
if $\sum_{i,j=1}^a |i\ran\lan j|\ot x_{i,j}\in (M_a\ot M_A)^+$ and
$\sum_{k,\ell=1}^b |k\ran\lan \ell |\ot y_{k,\ell}\in (M_b\ot M_B)^+$ then
$$
\sum_{i,j=1}^a \sum_{k,\ell=1}^b |i\ran |k\ran \lan j|\lan \ell |\ot \phi(x_{i,j}, y_{k,\ell})\in (M_{ab}\ot M_C)^+.
$$
This implies that the Choi matrix $C_\phi$ is positive because both  $\sum_{i,j=1}^a |i\ran\lan j|\ot |i\ran\lan j|$
and $\sum_{k,\ell=1}^b |k\ran\lan \ell |\ot |k\ran\lan \ell |$ are positive. Therefore, we see that (iv)
implies (v).
If $C_\phi$ is positive then it is the sum of rank one positive matrices by the spectral decomposition,
and so we see that $\phi$ is of the form in (vi) by the above discussion.
It is easy to see that the bi-linear map $\phi_V$ satisfies the condition (iii).
\end{proof}

The Hadamard product $[x_{i,j}]\circ [y_{i,j}]=[x_{i,j}y_{i,j}]$ between $n\times n$ matrices
is a typical example of a bi-linear map satisfying the conditions in
Theorem \ref{Choi-iso}. Its Choi matrix is given by
$$
\sum_{i,j=0}^{n-1} |i\ran\lan j|\ot  |i\ran\lan j|\ot  |i\ran\lan j|\in M_n\ot M_n\ot M_n,
$$
which is the rank one positive matrix onto the vector
$\sum_{i=0}^{n-1} |i\ran |i\ran |i\ran \in \mathbb C^{n^3}$.
We close this section with the linearization of
$(p,q,r)$-positive bi-linear maps.

\begin{theorem}\label{linearization}
Suppose that $\phi:{\mathcal S}\times{\mathcal T}\to{\mathcal
R}$ is a bi-linear map for operator systems $\mathcal S, \mathcal T$ and $\mathcal R$. Then the following are equivalent:
\begin{enumerate}
\item[(i)]
$\phi$ is $(p,q,r)$-positive;
\item[(ii)]
$\tilde \phi : \OMAX^p(\mathcal S) \otimes_{\max} \OMAX^q(\mathcal T)
\to \mathcal R$ is $r$-positive;
\item[(iii)]
$\tilde \phi : \OMAX^p(\mathcal S) \otimes_{\max} \OMAX^q(\mathcal T)
\to {\rm OMIN}^r (\mathcal R)$ is completely positive.
\end{enumerate}
\end{theorem}

\begin{proof}
The equivalence between (ii) and (iii) follows from \cite[Theorem 3.7]{X}.
Suppose that (ii) holds, and take $x \in M_p(\mathcal S)^+, y \in M_q(\mathcal T)^+$ and
$\alpha \in M_{r,pq}$. Since
$$
\alpha (x \otimes y) \alpha^* \in M_r(\OMAX^p(\mathcal S) \otimes_{\max} \OMAX^q(\mathcal T))^+,
$$
we have
$$
\alpha \phi_{p,q}(x,y) \alpha^* = \alpha \tilde{\phi}_{pq}(x \otimes y) \alpha^*
= \tilde{\phi}_{pq} (\alpha (x \otimes y) \alpha^*) \in M_r(\mathcal R)^+,
$$
and so, $\phi$ is $(p,q,r)$-positive.

For the direction (i) $\Longrightarrow$ (ii), we take
$z\in M_r( \OMAX^p (\mathcal S) \otimes_{\max} \OMAX^q(\mathcal T))^+$ and
arbitrary $\varepsilon_1, \varepsilon_2, \varepsilon_3 >0$.
By (\ref{max}), we can take
$x \in M_m(\OMAX^p(\mathcal S))^+, y \in M_n(\OMAX^q(\mathcal T))^+$ and $\alpha \in M_{r,mn}$
satisfying the relation
$$
z + \varepsilon_1 I_r \otimes 1_{\mathcal S} \otimes 1_{\mathcal T} = \alpha (x \otimes y) \alpha^*.
$$
By (\ref{max-k}), we may also find
$x_i \in M_p({\cl S})^+, y_j \in M_q({\cl T})^+$ and $\beta \in M_{m,ps}, \gamma \in M_{n,qt}$ satisfying
$$
x+ \varepsilon_2 I_m \otimes 1_{\mathcal S}=\beta \begin{pmatrix} x_1 && \\ &\ddots& \\ && x_s \end{pmatrix} \beta^*,
\qquad
y+\varepsilon_3 I_n \otimes 1_{\mathcal T}=\gamma \begin{pmatrix} y_1 && \\ &\ddots& \\ && y_t \end{pmatrix} \gamma^*.
$$
Write
$$
\alpha (\beta \otimes \gamma) =
\begin{pmatrix} \Theta_{(1,1)}& \cdots & \Theta_{(i,j)} & \cdots & \Theta_{(s,t)}
\end{pmatrix}\in M_{r,pqst}
$$
with $\Theta_{(i,j)} \in M_{r,pq}$.
Then, we have the identity
$$
\begin{aligned}
&\sum_{i=1}^s \sum_{j=1}^t \Theta_{(i,j)} \phi_{p,q}(x_i, y_j) \Theta_{(i,j)}^*  \\
& = \alpha(\beta \otimes \gamma)
\begin{pmatrix}
\phi_{p,q}(x_1, y_1) &&&& \\ &\ddots &&& \\ && \phi_{p,q}(x_i,y_j) && \\ &&& \ddots & \\ &&&& \phi_{p,q}(x_s,y_t)
\end{pmatrix}
(\beta \otimes \gamma)^* \alpha^* \\
& = \alpha \phi_{m,n}(x+\varepsilon_2 I_m \otimes 1_{\mathcal S}, y+\varepsilon_3 I_n \otimes 1_{\mathcal T}) \alpha^*
\end{aligned}
$$
which belongs to $M_r(\mathcal R)^+$ by $(p,q,r)$-positivity of $\phi$.
Expanding the last term, we have
$$
\begin{aligned}
& \alpha \phi_{m,n}(x, y) \alpha^* +
  \varepsilon_3 \alpha \phi_{m,n}(x,  I_n \otimes 1_{\mathcal T}) \alpha^*
  + \varepsilon_2 \alpha \phi_{m,n}(I_m \otimes 1_{\mathcal S}, y) \alpha^*\\
& \qquad + \varepsilon_2 \varepsilon_3 \alpha \phi_{m,n}(I_m \otimes 1_{\mathcal S}, I_n \otimes 1_{\mathcal T}) \alpha^*\\
&\le \tilde{\phi}_r(z)
  + (\varepsilon_1 \|\tilde{\phi}(1_{\mathcal S} \otimes 1_{\mathcal T}) \|
  + \varepsilon_3 \|  \alpha \phi_{m,n}(x,  I_n \otimes 1_{\mathcal T}) \alpha^* \| \\
&\qquad + \varepsilon_2 \|\alpha \phi_{m,n}(I_m \otimes 1_{\mathcal S}, y) \alpha^*\|
  + \varepsilon_2 \varepsilon_3 \|\alpha \phi_{m,n}(I_m \otimes 1_{\mathcal S}, I_n \otimes 1_{\mathcal T}) \alpha^*\|)
   I_r \otimes 1_{\mathcal R}
\end{aligned}
$$
We note that $x,y$ and $\alpha$ are independent of the choice of $\varepsilon_2$ and $\varepsilon_3$.
Therefore, we can conclude that
$\tilde{\phi}_r(z) \in M_r(\mathcal R)^+$ by the Archimedean property.
This proves that the linearization $\tilde\phi$ is $r$-positive.
\end{proof}

%%%%%%%%%%%%%%%%%%%%%%%%%%%%%%%%%%%%%%%%%%%%%%%%%%%%%%%%%%%%%%%%%%%%%%%%%%%%%%%%%%%%%%%%%%%%%%%%%%%%%%%%%%%%%%%%%%
%%%%%%%%%%%%%%%%%%%%%%%%%%%%%%%%%%%%%%%%%%%%%%%%%%%%%%%%%%%%%%%%%%%%%%%%%%%%%%%%%%%%%%%%%%%%%%%%%%%%%%%%%%%%%%%%%%
%%%%%%%%%%%%%%%%%%%%%%%%%%%%%%%%%%%%%%%%%%%%%%%%%%%%%%%%%%%%%%%%%%%%%%%%%%%%%%%%%%%%%%%%%%%%%%%%%%%%%%%%%%%%%%%%%%
%%%%%%%%%%%%%%%%%%%%%%%%%%%%%%%%%%%%%%%%%%%%%%%%%%%%%%%%%%%%%%%%%%%%%%%%%%%%%%%%%%%%%%%%%%%%%%%%%%%%%%%%%%%%%%%%%%
%%%%%%%%%%%%%%%%%%%%%%%%%%%%%%%%%%%%%%%%%%%%%%%%%%%%%%%%%%%%%%%%%%%%%%%%%%%%%%%%%%%%%%%%%%%%%%%%%%%%%%%%%%%%%%%%%%
\section{Schmidt numbers for tri-partite states}

We recall that the Schmidt rank of a vector
$\eta= \sum_{i=1}^n v_i \otimes w_i \in {\mathcal H}_B \otimes {\mathcal H}_C$ is equal to
the rank of the associate map $\lambda_\eta:{\mathcal H}_B\to{\mathcal H}_C$ given by
$$
\lambda_\eta (v) = \sum_{i=1}^n \langle \bar v_i |v \rangle  w_i.
$$
In fact, the correspondence $\eta\mapsto \lambda_\eta$ follows from the natural isomorphisms
$$
{\mathcal H}_B \otimes {\mathcal H}_C \simeq ({\mathcal H}_B)^* \otimes {\mathcal H}_C
\simeq \mathcal L ({\mathcal H}_B, {\mathcal H}_C).
$$
Here, $\mathcal H_B$ is self-dual because it has a canonical basis,
and the second isomorphism is the linearization of the bilinear map
$$
(f, w) \in  ({\mathcal H}_B)^* \times {\mathcal H}_C \mapsto f(\cdot) w \in \mathcal L ({\mathcal H}_B, {\mathcal H}_C).
$$ Since the map $\lambda_\eta$ follows from the above isomorphisms, it is independent of the tensor expression of $\eta$.

Applying the above isomorphism twice, we consider the natural isomorphisms
$$
\begin{aligned}
{\mathcal H}_A \otimes {\mathcal H}_B \otimes {\mathcal H}_C
&\simeq  ({\mathcal H}_A)^* \otimes ({\mathcal H}_B)^* \otimes {\mathcal H}_C\\
& \simeq  ({\mathcal H}_A)^* \otimes \mathcal L ({\mathcal H}_B, {\mathcal H}_C)
\simeq \mathcal L ({\mathcal H}_A, \mathcal L ({\mathcal H}_B, {\mathcal H}_C)),
\end{aligned}
$$
to get the analogous notion.
We write $\xi\in {\mathcal H}_A\ot{\mathcal H}_B\ot{\mathcal H}_C$ by
$$
\xi = \sum_{i=1}^n u_i \otimes \eta_i \in {\mathcal H}_A \otimes ({\mathcal H}_B \otimes {\mathcal H}_C)
$$
with $u_i \in {\mathcal H}_A$ and $\eta_i \in {\mathcal H}_B \otimes {\mathcal H}_C$.
Then, the above isomorphisms maps $\xi$ to the linear map $\Lambda_\xi : {\mathcal H}_A \to \mathcal L({\mathcal H}_B, {\mathcal H}_C)$ given by
$$
\Lambda_\xi  (u)= \sum_{i=1}^n \langle \bar u_i |u \rangle  \lambda_{\eta_i}.
$$
Now, we consider the following three numbers:
$$
\begin{aligned}
\alpha_\xi &= {\rm rank} \Lambda_\xi,\\
\beta_\xi &=\dim \bigvee \{ {\rm supp}\ T : T \in {\rm ran} \Lambda_\xi \}\\
\gamma_\xi &=\dim \bigvee \{  {\rm ran} \ T : T \in {\rm ran} \Lambda_\xi \}.
\end{aligned}
$$
Here, the support of a linear map means the orthogonal complement of its kernel.
Since the map $\Lambda_\xi$ follows from the above isomorphisms, the map $\Lambda_\xi$
and three numbers $\alpha_\xi, \beta_\xi, \gamma_\xi$ are independent of the tensor expression of $\xi$.

\begin{definition}
\emph{We call the triplet $(\alpha_\xi,\beta_\xi,\gamma_\xi)$ the} Schmidt rank \emph{of the vector
$\xi\in{\mathcal H}_A\ot{\mathcal H}_B\ot {\mathcal H}_C$ and write $\SR(\xi)=(\alpha_\xi,\beta_\xi,\gamma_\xi)$}.
\end{definition}

\begin{theorem}\label{SR}
Suppose that $\xi \in {\mathcal H}_A \otimes {\mathcal H}_B \otimes {\mathcal H}_C$ and
$1 \le p \le a, 1 \le q \le b, 1 \le r \le c$. Then the following are equivalent:
\begin{enumerate}
\item[(i)]
$\alpha_\xi\le p$, $\beta_\xi\le q$, $\gamma_\xi\le r$;
\item[(ii)]
there exist vectors
$\{ u_i \}_{i=1}^p \subset {\mathcal H}_A$,
$\{ v_j \}_{j=1}^q \subset {\mathcal H}_B$,
$\{ w_k \}_{k=1}^r \subset {\mathcal H}_C$ and scalars
$c_{i,j,k}$ such that
$$
\xi = \sum_{i=1}^p \sum_{j=1}^q \sum_{k=1}^r c_{i,j,k} u_i \otimes v_j \otimes w_k;
$$
\item[(iii)]
there exist orthonormal vectors
$\{ u_i \}_{i=1}^p \subset {\mathcal H}_A$,
$\{ v_j \}_{j=1}^q \subset {\mathcal H}_B$,
$\{ w_k \}_{k=1}^r \subset {\mathcal H}_C$ and scalars
$c_{i,j,k}$ such that
$$
\xi = \sum_{i=1}^p \sum_{j=1}^q \sum_{k=1}^r c_{i,j,k} u_i \otimes v_j \otimes w_k.
$$
\end{enumerate}
\end{theorem}

\begin{proof}
(ii) $\Rightarrow$ (i). Suppose that $\xi$ is given as in the statement (ii). Then we have
$$
\xi=\sum_{i=1}^p u_i\ot
\left(\sum_{j=1}^q\sum_{k=1}^r c_{i,j,k}v_j\ot w_k\right)
=\sum_{i=1}^p u_i\ot \left(\sum_{j=1}^q v_j\ot
\left(\sum_{k=1}^r c_{i,j,k} w_k\right)\right).
$$
Therefore, we have the following relation
\begin{equation}\label{zzz}
\Lambda_\xi (u) (v)= \sum_{i=1}^p \langle \bar u_i |u \rangle \lambda_{\sum_{j=1}^q v_j\ot
\left(\sum_{k=1}^r c_{i,j,k} w_k\right)}(v)
= \sum_{i=1}^p \sum_{j=1}^{q} \sum_{k=1}^r c_{i,j,k} \langle  \bar u_i |u \rangle \langle \bar v_j | v \rangle w_k,
\end{equation}
for $u\in{\mathcal H}_A$ and $v\in{\mathcal H}_B$.
This relation tells us that  $\bigvee \{  {\rm ran} \ T : T \in {\rm ran} \Lambda_\xi \}$ is a
subspace of ${\rm span}\{ w_k : 1 \le k \le r \}$. Therefore, we have $\gamma_\xi\le r$.
For the inequality $\alpha_\xi\le p$, it suffices to show that ${\rm supp} \Lambda_\xi$
is a subspace of ${\rm span} \{ \bar u_i : 1 \le i \le p\}$,
or equivalently
$$
\lan \bar u_i|u\ran =0, \ i=1,2,\dots,p\ \Longrightarrow u\in {\rm Ker} \Lambda_\xi,
$$
which follows from (\ref{zzz}).
It remains to show $\beta_\xi\le q$.
If $\langle  \bar v_j |v \rangle=0$ for each $j=1,2,\dots, q$,
then we have $\Lambda_\xi (u) (v)= 0$ by (\ref{zzz}). This means that
$\{\bar v_j : 1 \le j \le q \}^\perp$ is the subspace of
$\bigcap \{ \ker T : T \in {\rm ran} \Lambda_\xi \}$.
Considering their orthogonal complements, we see that the space
$\bigvee \{ {\rm supp}\ T : T \in {\rm ran} \Lambda_\xi \}$ is a subspace of
${\rm span}\{ \bar v_j : 1 \le j \le q \}$.

(i) $\Rightarrow$ (iii).
We write $\alpha_\xi=\alpha, \beta_\xi=\beta, \gamma_\xi=\gamma$.
We take an orthonormal basis $\{ \bar u_i \}_{i=1}^\alpha$ of the support of $\Lambda_\xi$,
an orthonormal basis $\{ \bar v_j \}_{j=1}^\beta$ of $\bigvee \{ {\rm supp}\ T : T \in {\rm ran} \Lambda_\xi \}$
and an orthonormal basis $\{ w_k \}_{k=1}^\gamma$ of $\bigvee \{  {\rm ran} \ T : T \in {\rm ran} \Lambda_\xi \}$.
There exist scalars $c_{i,j,k}$ such that
$\Lambda_\xi (\bar u_i) (\bar v_j)=\sum_{k=1}^\gamma c_{i,j,k} w_k$.
It suffices to show
$\xi = \sum_{i=1}^\alpha \sum_{j=1}^\beta \sum_{k=1}^\gamma c_{i,j,k} u_i \otimes v_j \otimes w_k$.
To do this, we show that
$$
\Lambda_\xi=\Lambda_{\sum_{i=1}^\alpha \sum_{j=1}^\beta \sum_{k=1}^\gamma c_{i,j,k} u_i \otimes v_j \otimes w_k}.
$$
First of all, we have
$$
\Lambda_\xi (\bar u_{i_0})(v)
= \Lambda_\xi (\bar u_{i_0})\left(\sum_{j=1}^\beta \langle \bar v_j | v \rangle \bar v_j +v' \right) \\
= \sum_{j=1}^\beta \sum_{k=1}^\gamma c_{i_0, j,k} \langle \bar v_j | v \rangle w_k
$$
for $v \in \HA$ with $v' \perp \bar v_j$. On the other hand, we have
$$
\Lambda_{\sum_{i=1}^\alpha \sum_{j=1}^\beta \sum_{k=1}^\gamma c_{i,j,k} u_i \otimes v_j \otimes w_k}(\bar u_{i_0})(v)
=\sum_{j=1}^\beta\sum_{k=1}^\gamma c_{i_0, j,k} \lambda_{v_j\ot w_k}(v)
=\sum_{j=1}^\beta \sum_{k=1}^\gamma c_{i_0, j,k} \langle \bar v_j | v \rangle w_k.
$$
For all $u\in \HA$ with $u \perp \bar u_i$, we have
$$
\Lambda_\xi (u) = 0
= \sum_{i=1}^\alpha \langle \bar u_i | u \rangle \lambda_{\sum_{j=1}^\beta \sum_{k=1}^\gamma c_{i,j,k} v_j \otimes w_k}
= \Lambda_{\sum_{i=1}^\alpha \sum_{j=1}^\beta \sum_{k=1}^\gamma c_{i,j,k} u_i \otimes v_j \otimes w_k}(u).
$$
If we put $c_{i,j,k}=0$ if $\alpha < i \le p$ or $\beta < j \le q$ or $\gamma < k \le r$, we get the expression.
\end{proof}

For a vector $\xi\in \HA\ot \HB\ot \HC$, we write
$\SR(\xi)\le (p,q,r)$ if the conditions in Theorem \ref{SR} are satisfied.
For a permutation $\sigma$ in $\{A,B,C\}$, we denote by $\xi^\sigma$ the vector in
${\mathcal H}_{\sigma A} \otimes {\mathcal H}_{\sigma B} \otimes {\mathcal H}_{\sigma C}$
obtained by the flip operator under $\sigma$.
For a given triplet $S=(s_A,s_B,s_C)$ and a permutation $\sigma$, we denote by
$S^\sigma$ the triplet $(s_{\sigma A}, s_{\sigma _B}, s_{\sigma C})$.

\begin{corollary}\label{permutation}
$\SR(\xi)= (\alpha,\beta,\gamma)$ if and only if $\SR(\xi^\sigma)= (\alpha,\beta,\gamma)^\sigma$
\end{corollary}

\begin{proof}
By Theorem \ref{SR}, we see that
$\SR(\xi)\le (p,q,r)$ if and only if $\SR(\xi^\sigma)\le (p,q,r)^\sigma$.
We apply this for the cases $p=\alpha, \alpha-1$ and $q=\beta, \beta-1$ and $r=\gamma, \gamma-1$
to get the conclusion.
\end{proof}

In the case of tensor product $\HA\ot\HB$ of two spaces, the Schmidt rank of a vector must fall down in $\{1,2,\dots, a\meet b\}$.
We show that all the possible combinations of triplets for $\SR(\xi)$ is given by the set
$$
\Sigma_{a,b,c}=
\{ (\alpha, \beta, \gamma) \in \mathbb N^3 : \alpha \le \beta \gamma, \beta \le \gamma \alpha, \gamma \le \alpha \beta,
1 \le \alpha \le a, 1 \le \beta \le b, 1 \le \gamma \le c \}.
$$

\begin{proposition}\label{restriction}
There exists $\xi\in\HA\ot\HB\ot\HC$ such that $\SR(\xi)=(\alpha,\beta,\gamma)$ if and only if
$(\alpha,\beta,\gamma)\in \Sigma_{a,b,c}$.
\end{proposition}

\begin{proof}
We may assume that $\alpha\le \beta \le \gamma$ by Corollary \ref{permutation}.

($\Longrightarrow$)
Let $\xi\in\HA\ot\HB\ot\HC$ with $\SR(\xi)=(\alpha, \beta, \gamma)$.
We take a basis
$\{ u_1, \cdots, u_\alpha \}$ of ${\rm supp} \Lambda_\xi$. Then, we have
$$
\begin{aligned}
\gamma  = \dim ({\rm ran} \Lambda_\xi(u_1) \vee \cdots \vee {\rm ran} \Lambda_\xi(u_\alpha))
&\le \sum_{i=1}^\alpha \dim {\rm ran}\Lambda_\xi (u_i) \\
&= \sum_{i=1}^\alpha \dim {\rm supp} \Lambda_\xi (u_i)
\le \alpha \beta.
\end{aligned}
$$

($\Longleftarrow$)
Let $\gamma = \beta \cdot k +r$ with $0 \le r < \beta$. Since $\gamma \le \alpha \beta$,
we have $k \le \alpha$. We take orthonormal sets $\{ u_i \}_{i=1}^\alpha$ in
$\HA$, $\{ v_j \}_{j=1}^\beta$ in $\HB$ and $\{ w_k \}_{k=1}^\gamma$ in
$\HC$. We consider the vector
$$
\xi = \sum_{i=1}^k u_i \otimes \left(\sum_{j=1}^\beta v_j \otimes w_{(i-1)\beta +j}\right)
+ u_{k+1} \otimes \left(\sum_{j=1}^r v_j \otimes w_{k\beta +j}\right)
+ \sum_{i=1}^{\alpha -k-1} u_{k+1+i} \otimes v_i \otimes w_i.
$$
in $\HA \ot \HB \ot \HC$.
The expression in the last term is legitimate because we assume $\alpha\le\beta\le\gamma$.
If $\gamma = \alpha \beta$, then $k=\alpha$,
so we ignore the last two terms. If $(\alpha-1)\beta \le \gamma < \alpha \beta$, then $k=\alpha-1$, so we ignore the last term.
The ranges of $\Lambda_\xi (\bar u_i)$ for $1 \le i \le k+1$ are orthogonal. The ranges of
$\Lambda_\xi (\bar u_i)$ for $k+2 \le i \le \alpha$ are orthogonal and their join is a proper subspace of
${\rm ran} \Lambda_\xi (\bar u_1)$.
Hence, we conclude that the set $\{ \Lambda_\xi (\bar u_i) : 1 \le i \le \alpha \}$ is linearly independent,
and so $\alpha = \alpha_\xi$.

Since ${\rm supp} \Lambda_\xi (\bar u_i) = {\rm span} \{ \bar v_j : 1 \le j \le \beta \}$ for $1 \le i \le k$ and
${\rm supp} \Lambda_\xi (\bar u_i)$ is a subspace of ${\rm span} \{ \bar v_j  : 1 \le j \le \beta \}$ for $k+1 \le i \le \alpha$,
we have $\beta = \beta_\xi$.
Finally, the join of the ranges of $\Lambda_\xi (\bar u_i)$ for $1 \le i \le \alpha$ is ${\rm span} \{ w_k : 1 \le k \le \gamma \}$,
and so $\gamma = \gamma_\xi$.
\end{proof}

\begin{corollary}\label{equal}
Let $\xi\in\HA\ot\HB\ot\HC$ with $\SR(\xi)=(\alpha,\beta,\gamma)$. If one of $\alpha, \beta, \gamma$ is $1$, then the other two are equal.
\end{corollary}

We recall \cite{eom-kye} that the dual cone $\mathbb V_s$ of the convex cone of all $s$-positive linear maps may be described in terms of
Schmidt ranks. More precisely, $\mathbb V_s$ consists of all bi-partite unnormalized states $\varrho$ which can be expressed
by $\varrho=\sum_i |\xi_i\ran\lan \xi_i|$, where Schmidt rank of $\xi_i$ is less than or equal to $s$.
A bi-partite state $\varrho$ is separable if and only if it belongs to $\mathbb V_1$ by definition.
If $\varrho$ belongs to $\mathbb V_s$ but does not belong to $\mathbb V_{s-1}$ then we say that $\varrho$
has Schmidt number $s$.
This motivates the following:

\begin{definition}
\emph{For a tri-partite unnormalized state $\varrho\in M_A\ot M_B\ot M_C$,
we write $\SN(\varrho) \le (p, q, r)$ if there exist
$\xi_1, \cdots, \xi_n \in \HA \otimes {\mathcal H}_B \otimes {\mathcal H}_C$ such that
$\varrho = \sum_{i=1}^n |\xi_i\ran\lan \xi_i|$ and
$\SR(\xi_i)\le (p,q,r)$ for each $i=1,2,\dots,n$.
We denote by $\mathbb S_{p,q,r}$ the set of all tri-partite unnormalized states $\varrho$ in $M_A\ot M_B\ot M_C$ with the property
$\SN(\varrho) \le (p, q, r)$.}
\end{definition}

The set $\mathbb S_{p,q,r}$ is, by definition, the cone generated by the set $\mathcal E$ of all rank one projections onto unit vectors $\xi$ with
$\SR(\xi)\le (p,q,r)$. This is equivalent to satisfying the conditions
$$
\SR(\xi)\le (p,b,c),\quad
\SR(\xi)\le (a,q,c),\quad
\SR(\xi)\le (a,b,r),
$$
simultaneously. The first condition tells us that $\rk \xi\le p$ with respect to the $A$-$BC$ bi-partition of
$\xi\in \HA\ot (\HB\ot \HC)$, and so  the set of all unit vectors $\xi$ with $\SR(\xi)\le (p,b,c)$ is compact, and same for
the other two conditions by Corollary \ref{permutation}. Therefore,
we see that the set of all unit vectors $\xi$ with $\SR(\xi)\le (p,q,r)$ is compact. Since the map
$| \xi \rangle \mapsto |\xi \rangle \langle \xi|$ is continuous,
the set $\mathcal E$ is also compact. Hence, we conclude that the set of states in $\mathbb S_{p,q,r}$ is also compact, by
Caratheodory's theorem \cite[Theorem 17.2]{R}.

Theorem \ref{SR} tells us that
$\SN(\varrho)\le (1,1,1)$ if and only if $\varrho$ is the convex sum of pure product states,
which is noting but the definition of (genuine) separability.
For a state $\varrho\in M_A\ot M_B\ot M_C$, it is also clear that
$\SN(\varrho)\le (1,b,c)$ if and only if it is $A$-$BC$ separable, that is, it is
separable as a bi-partite state in $M_A\ot (M_B\ot M_C)$.
We also see that $\SN(\varrho)\le (a,1,c)$ if and only if it is $B$-$CA$ separable,
and $\SN(\varrho)\le (a,b,1)$ if and only if it is $C$-$AB$ separable.

The property $\SN(\xi)\le (p,q,r)$ can be described in
terms of positivity in the maximal tensor product of super maximal
operator systems.

\begin{theorem}\label{SN}
An unnormalized state $\rho \in M_A \otimes M_B \otimes M_C$ belongs to $\mathbb S_{p,q,r}$ if and only if
we have
$$
\varrho \in [\OMAX^p(M_A)  \otimes_{\max} \OMAX^q(M_B) \otimes_{\max} \OMAX^r (M_C)]^+.
$$
\end{theorem}

\begin{proof}
($\Longleftarrow$)
Suppose that $\varrho$ belongs to the positive cone in the statement, and take arbitrary
$\varepsilon_i >0$ for $i=1,2,3,4,5$.
By (\ref{max}), we can write
$$
\varrho + \varepsilon_1 1 = \alpha (X \otimes X') \alpha^*
$$
for $X \in M_\ell(\OMAX^p(M_A))^+, X' \in M_{\ell^\prime}(\OMAX^q(M_B) \otimes_{\max} M_n(\OMAX^r(M_C))^+$ and $\alpha \in M_{1,\ell\ell^\prime}$,
and
$$
X'+\varepsilon_2 I_{\ell'} \otimes 1 = \alpha'(Y \otimes Z) \alpha'^*
$$
for $Y \in M_m(\OMAX^q(M_B))^+, Z \in M_n(\OMAX^r(M_C))^+$ and $\alpha' \in M_{\ell^\prime,mn}.$
Moreover, each $X, Y, Z$ can be written as
$$
\begin{aligned}
X+\varepsilon_3 I_{\ell} \otimes 1 &=\beta \diag (X_1,\dots, X_s) \beta^*,\\
Y+\varepsilon_4 I_m \otimes 1 &=\gamma \diag (Y_1,\dots, Y_t) \gamma^*,\\
Z+\varepsilon_5 I_n \otimes 1 &=\delta \diag (Z_1,\dots, Z_u) \delta^*,
\end{aligned}
$$
for $X_i \in M_p(M_A)^+, Y_j \in M_q(M_B)^+, Z_k \in M_r(M_C)^+$, and
scalar matrices $\beta \in M_{l,ps}, \gamma \in M_{m,qt}$ and $\delta \in M_{n,ru}$, by (\ref{max-k}).
Here, we denote by $\diag (A_1,\dots, A_n)$ for the $n\times n$ block diagonal matrix with
diagonal entries $A_1,\dots, A_n$.

Combining them and putting $\Theta :=\alpha (I_\ell \otimes \alpha')(\beta \otimes \gamma \otimes \delta) \in M_{1,pqrstu}$,
we see that the set $\Omega$ consisting
$$
\Theta \left(\diag(X_1,\dots,X_s)\ot \diag (Y_1,\dots, Y_t)\ot \diag (Z_1,\dots,Z_u)\right)\Theta^*\in M_A\ot M_B\ot M_C
$$
through
$\Theta \in M_{1,pqrstu},\ X_i \in M_p(M_A)^+,\ Y_j \in M_q(M_B)^+,\ Z_k \in M_r(M_C)^+$ and $ s,t,u =1,2,\dots$
is dense in
$$
[\OMAX^p(M_A)  \otimes_{\max} \OMAX^q(M_B) \otimes_{\max} \OMAX^r (M_C)]^+.
$$
Because the set of states in $\mathbb S_{p,q,r}$ is compact, it is enough to show that each $\rho \in \Omega$ satisfies
$\SN(\rho) \le (p,q,r)$ by normalization.

We write $\varrho \in \Omega$ as the summation
$$
\rho = \sum_{i=1}^s \sum_{j=1}^t \sum_{k=1}^u \Theta
\diag (0,\dots, 0,X_i \otimes Y_j \otimes Z_k ,0,\dots,0)\Theta^*
$$
and consider only $(i_0, j_0, k_0)$-th term.
We let
$U_i \in \mathbb C^p \otimes {\mathcal H}_A$ be the $i$-th column of
$X_{i_0}^{1 \slash 2}$, $V_j \in \mathbb C^q \otimes {\mathcal H}_B$
the $j$-th column of $Y_{j_0}^{1 \slash 2}$ and
$W_k \in \mathbb C^r \otimes {\mathcal H}_C$ the $k$-th column of
$Z_{k_0}^{1 \slash 2}$ ($1 \le i \le pa, 1 \le j \le qb, 1 \le k \le rc$).
Then, we have
$$
\begin{aligned}
&  \Theta
  \diag (0,\dots,0, X_{i_0} \otimes Y_{j_0} \otimes Z_{k_0},0,\dots,0)
  \Theta^* \\
= & \sum_{i=1}^{pa} \sum_{j=1}^{qb} \sum_{k=1}^{rc}  \Theta
  \diag  (0,\dots,0,U_i U_i^* \otimes V_j V_j^* \otimes W_k \ W_k^*,0,\dots,0)
  \Theta^* \\ =
& \sum_{i=1}^{pa} \sum_{j=1}^{qb} \sum_{k=1}^{rc}
\left[\Theta
  \begin{pmatrix} 0 \\ \vdots \\ U_i \otimes V_j \otimes W_k \\ \vdots \\ 0 \end{pmatrix})\right]
  \left[\Theta
  \begin{pmatrix} 0 \\ \vdots \\ U_i \otimes V_j \otimes W_k \\ \vdots \\ 0 \end{pmatrix})\right]^*.
\end{aligned}
$$
We write
$U_i = (u_1, \cdots, u_p)^t \in \mathbb C^p \otimes {\mathcal H}_A$,
$V_j = (v_1, \cdots, v_q)^t \in \mathbb C^q \otimes {\mathcal H}_B$
 and $W_k = (w_1, \cdots, w_r)^t \in \mathbb C^r \otimes {\mathcal H}_C$.
 Then, the term in the square bracket is the linear combination of
 $\{ u_i \otimes v_j \otimes w_k : 1 \le i \le p, 1 \le j \le q, 1 \le k \le r \}$.
 By Theorem \ref{SR}, we see that $\SN(\varrho)\le (p,q,r)$.
The converse is merely the reverse of the above argument.
\end{proof}

The proof of Theorem \ref{SN} actually shows that the relation
$$
\Omega =  [\OMAX^p(M_A)  \otimes_{\max} \OMAX^q(M_B) \otimes_{\max} \OMAX^r (M_C)]^+
$$
holds.

So far, we have focused on tri-partite states, in order to avoid the excessive notations.
Many parts in this section can be extended for multi-partite states.
By the isomorphisms
$$
{\cl H}_1\ot {\cl H}_2\ot \cdots {\cl H}_{n-1}\ot {\cl H}_n
\simeq
{\cl H}^*_1\ot {\cl H}^*_2\ot \cdots {\cl H}^*_{n-1}\ot {\cl H}_n
\simeq
{\cl L}({\cl H}_1, {\cl L}({\cl H}_2,\dots, {\cl L}({\cl H}_{n-1},{\cl H}_n))
$$
for finite dimensional Hilbert spaces $\cl H_i = \mathbb C^{d_i}$, we associate
each $\xi$ in $\bigotimes_{i=1}^n \mathcal H_i$ with the linear map
$\Lambda_\xi : \mathcal H_1 \to  {\cl L}({\cl H}_2,\dots, {\cl L}({\cl H}_{n-1},{\cl H}_n))$.
Now, we consider the following $n$ ($\ge 3$) numbers:
$$
\begin{aligned}
\alpha_\xi^1 &=\dim {\rm supp} \Lambda_\xi\\
\alpha_\xi^2 &=\dim \bigvee \{ {\rm supp}\ T_1 : T_1 \in {\rm ran} \Lambda_\xi \}\\
\alpha_\xi^3 &=\dim \bigvee \{ {\rm supp}\ T_2 : T_2 \in {\rm ran} T_1, T_1 \in {\rm ran} \Lambda_\xi \}\\
\vdots\\
\alpha_\xi^{n-1} &=\dim \bigvee \{ {\rm supp}\ T_{n-2} : T_{n-2} \in {\rm ran} T_{n-3}, \dots, T_2
        \in {\rm ran} T_1, T_1 \in {\rm ran} \Lambda_\xi \}\\
\alpha_\xi^n &=\dim \bigvee \{  {\rm ran}\ T_{n-2} : T_{n-2} \in {\rm ran} T_{n-3}, \dots, T_2
           \in {\rm ran} T_1, T_1 \in {\rm ran} \Lambda_\xi \}.
\end{aligned}
$$
When $n=2$, it is natural to define $\alpha_\xi^2 = \dim {\rm ran} \Lambda_\xi$ in the above context.
We define the Schmidt rank of $\xi$ by the $n$-tuple $(\alpha_\xi^1, \dots, \alpha_\xi^n)$ and denote it by $SR(\xi)$. We have
$$
\alpha_\xi^1 \le s_1, \dots, \alpha_\xi^n \le s_n
$$
if and only if
there exist vectors
$\{ u_{\iota_1}^1 \}_{\iota_1=1}^{s_1} \subset {\mathcal H_1}, \dots, \{ u_{\iota_n}^n \}_{\iota_n=1}^{s_n} \subset \mathcal H_n$ and scalars
$c_{\iota_1, \dots, \iota_n}$ such that
$$
\xi = \sum_{\iota_n=1}^{s_n} \dots \sum_{\iota_1=1}^{s_1} c_{\iota_1, \dots, \iota_n} u_{\iota_1}^1 \otimes \cdots \otimes u_{\iota_n}^n.
$$
In this case, we write $SR(\xi) \le (s_1, \dots, s_n)$.

For a multi-partite unnormalized state $\varrho\in \bigotimes_{i=1}^n \mathcal L (\mathcal H_i)$,
we write $\SN(\varrho) \le (s_1, \dots, s_n)$ if there exist
$\xi_1, \dots, \xi_m \in \bigotimes_{i=1}^n \mathcal H_i$ such that
$\varrho = \sum_{i=1}^m |\xi_i\ran\lan \xi_i|$ and
$\SR(\xi_i)\le (s_1, \dots, s_n)$ for each $i=1,2,\dots,m$.
We denote by $\mathbb S_{s_1, \dots,s_n}$ the set of all multi-partite
unnormalized states $\varrho$ in $\bigotimes_{i=1}^n \mathcal L (\mathcal H_i)$ with the property
$\SN(\varrho) \le (s_1, \dots, s_n)$. Then, we have
$$
\mathbb S_{s_1, \dots,s_n} = [\OMAX^{s_1}(\cl L (\cl H_1))  \otimes_{\max} \cdots \otimes_{\max} \OMAX^{s_n} (\cl L (\cl H_n))]^+.
$$

In the bi-partite case, the Schumidt number of $\varrho \in M_m \otimes M_n$ is less than or equal to
$k$ in the usual sense if and only if $\SN(\varrho) \le (k,k)$ if and only if
$\SN(\varrho) \le (m,k)$ because  $\SR(\xi) = (j,k)$ cannot occur for $j \ne k$
due to the first isomorphism theorem. By \cite[Theorem 4.6]{X}, we have $M_m = \OMAX^m (M_m)$.
Therefore, the Schmidt number of $\varrho$ is less than or
equal to $k$ in the usual sense if and only if
$$
\varrho \in [\OMAX^m (M_m) \otimes_{\max} \OMAX^k(M_n)]^+ = M_m(\OMAX^k(M_n))^+,
$$
which was proved in \cite[Theorem 5]{JKPP}.

%%%%%%%%%%%%%%%%%%%%%%%%%%%%%%%%%%%%%%%%%%%%%%%%%%%%%%%%%%%%%%%%%%%%%%%%%%%%%%%%%%%%%%%%%%%%%%%%%%%%%%%%%%%%%%%%%%
%%%%%%%%%%%%%%%%%%%%%%%%%%%%%%%%%%%%%%%%%%%%%%%%%%%%%%%%%%%%%%%%%%%%%%%%%%%%%%%%%%%%%%%%%%%%%%%%%%%%%%%%%%%%%%%%%%
%%%%%%%%%%%%%%%%%%%%%%%%%%%%%%%%%%%%%%%%%%%%%%%%%%%%%%%%%%%%%%%%%%%%%%%%%%%%%%%%%%%%%%%%%%%%%%%%%%%%%%%%%%%%%%%%%%
%%%%%%%%%%%%%%%%%%%%%%%%%%%%%%%%%%%%%%%%%%%%%%%%%%%%%%%%%%%%%%%%%%%%%%%%%%%%%%%%%%%%%%%%%%%%%%%%%%%%%%%%%%%%%%%%%%
%%%%%%%%%%%%%%%%%%%%%%%%%%%%%%%%%%%%%%%%%%%%%%%%%%%%%%%%%%%%%%%%%%%%%%%%%%%%%%%%%%%%%%%%%%%%%%%%%%%%%%%%%%%%%%%%%%
\section{Duality}

For $\varrho\in M_A\ot M_B\ot M_C$ and a bi-linear map $\phi:M_A\times M_B\to M_C$, the bi-linear pairing
$\lan \varrho,\phi\ran$ is defined \cite{kye_3_qubit} by
\begin{equation}\label{def-dual}
\lan\varrho,\phi\ran = \lan \varrho,C_\phi\ran = \tr (C_\phi \varrho^\ttt).
\end{equation}
If $\varrho= u\ot v\ot w\in M_A\ot M_B\ot M_C$, then the pairing is given by
$$
\lan u\ot v\ot w,\phi\ran
= \lan \phi(u,v),w\ran = \tr(\phi(u,v)w^\ttt).
$$
We denote by $\mathbb P_{p,q,r}$ the convex cone consisting of all
$(p,q,r)$-positive bi-linear maps from $M_A\times M_B$ to $M_C$.

By (\ref{duality-max}) and (\ref{duality-super}), we have the complete order isomorphisms
$$
\begin{aligned}
& (\OMAX^p(M_A) \otimes_{\max} OMAX^q (M_B) \otimes_{\max} \OMAX^r(M_C)))^* \\
\simeq & \mathcal L (\OMAX^p(M_A) \otimes_{\max} \OMAX^q(M_B), \OMAX^r (M_C)^*) \\
\simeq & \mathcal L (\OMAX^p(M_A) \otimes_{\max} \OMAX^q(M_B), {\rm OMIN}^r (M_C)).
\end{aligned}
$$
In these isomorphisms, a functional $\varphi$ on $\OMAX^p (M_A) \otimes_{\max} \OMAX^q(M_B) \otimes_{\max} \OMAX^r(M_C)$
is positive if and only if its associated linear map
$$
\tilde{\phi} : \OMAX^p(M_A) \otimes_{\max} \OMAX^q(M_B) \to {\rm OMIN}^r(M_C)
$$
is completely positive. By Theorem \ref{linearization} and Theorem \ref{SN}, we have the following duality between
the convex cones $\mathbb S_{p,q,r}$ and $\mathbb P_{p,q,r}$. We give here a direct elementary proof.

\begin{theorem}\label{dual}
The convex cones $\mathbb S_{p,q,r}$ and $\mathbb P_{p,q,r}$ are dual to each other.
In other words, $\varrho\in\mathbb S_{p,q,r}$ if and only if
$\lan\varrho,\phi\ran\ge 0$ for each $\phi\in\mathbb P_{p,q,r}$.
\end{theorem}

\begin{proof}
By the separation theorem for a point outside of a closed convex set, it is sufficient to show that a bilinear map
$\phi : M_A \otimes M_B \to M_C$ is $(p,q,r)$-positive if and only if
$\langle \phi, \varrho \rangle \ge 0$ for all $\rho \in \mathbb S_{p,q,r}$.
Suppose that $\phi : M_A \otimes M_B \to M_C$ is $(p,q,r)$-positive and $\xi \in \HA \ot \HB \ot \HC$ with $\SR(\xi) \le (p,q,r)$.
By Theorem \ref{SR}, there exist vectors
$\{ u_i \}_{i=1}^p \subset {\mathcal H}_A$,
$\{ v_j \}_{j=1}^q \subset {\mathcal H}_B$,
$\{ w_k \}_{k=1}^r \subset {\mathcal H}_C$ and scalars
$c_{i,j,k}$ such that
$$
\xi = \sum_{i=1}^p \sum_{j=1}^q \sum_{k=1}^r c_{i,j,k} u_i \otimes v_j \otimes w_k.
$$
Then we have
$$
\begin{aligned}
 \langle \phi, \ \xi \xi^* \rangle
= & \left\langle \phi, \sum_{i,\ell=1}^p \sum_{j,m=1}^q
    \sum_{k,n=1}^r c_{i,j,k} \bar c_{\ell,m,n} u_i u_\ell^* \otimes v_j v_m^* \otimes w_k w_n^* \right\rangle \\
= & \sum_{i,\ell=1}^p \sum_{j,m=1}^q
    \sum_{k,n=1}^r  c_{i,j,k} \bar c_{\ell,m,n} \langle \tilde{\phi} (u_i u_\ell^* \otimes v_j v_m^*), w_k w_n^* \rangle_{M_C} \\
= & \left\langle  \left[\sum_{i,\ell=1}^p \sum_{j,m=1}^q c_{i,j,k} \bar c_{\ell,m,n}
    \tilde{\phi} (u_i u_\ell^* \otimes v_j v_m^*)\right]_{k,n}, [w_k w_n^* ]_{k,n} \right\rangle_{M_r(M_C)}.
\end{aligned}
$$
If we denote by $\alpha$ the  $r \times pq$ matrix whose entries are given by $\alpha_{k,(i,j)}=c_{i,j,k}$, then
the above quantity coincides with the following:
$$
 \left\langle \alpha \phi_{p,q} \left( \begin{pmatrix} u_1 \\ \vdots \\ u_p \end{pmatrix}
    \begin{pmatrix} u_1^*& \cdots & u_p^* \end{pmatrix},
    \begin{pmatrix} v_1 \\ \vdots \\ v_q \end{pmatrix}
    \begin{pmatrix} v_1^*& \cdots & v_q^* \end{pmatrix}\right) \alpha^*,
    \begin{pmatrix} w_1 \\ \vdots \\ w_r \end{pmatrix}
    \begin{pmatrix} w_1^*& \cdots & w_r^* \end{pmatrix} \right\rangle_{M_r(M_C)},
$$
which is nonnegative, because $\phi$ is $(p,q,r)$-positive.
Since every positive matrix is the sum of rank one positive matrices,
we get the converse by reversing of the above argument.
\end{proof}

Recall that a Hermitian matrix $W\in M_A\ot M_B$ is said to be $s$-block positive
if $\langle W, \varrho \rangle \ge 0$ for each $\varrho$ whose Schmidt number is less than or equal to $s$.
In this context, it is reasonable to
say that
a Hermitian matrix $W \in M_A \otimes M_B \otimes M_C$ is $(p,q,r)$-block positive if
$\langle W, \varrho \rangle \ge 0$ for all $\varrho \in \mathbb S_{p,q,r}$.
Then Theorem \ref{dual} and (\ref{duality-super}) together with \cite[Proposition 1.16]{FP} tells us the following.

\begin{corollary}\label{ChoiBlock}
Suppose that $\phi : M_A \times M_B \to M_C$ is a bi-linear map. Then the following are equivalent:
\begin{enumerate}
\item[(i)] $\phi$ is $(p,q,r)$-positive;
\item[(ii)] the Choi matrix $C_\phi \in M_A \otimes M_B \otimes M_C$ is $(p,q,r)$-block positive;
\item[(iii)] $C_\phi$ belongs to $(\OMIN^p(M_A) \otimes_{\min} \OMIN^q (M_B) \otimes_{\min} \OMIN^r(M_C))^+$.
\end{enumerate}
\end{corollary}

A tri-partite state $\varrho$ is called bi-separable if it is in the convex hull of
all $A$-$BC$, $B$-$CA$ and $C$-$AB$ separable states. This happens if and only if
$\varrho$ belongs to the convex hull of the convex sets
$\mathbb S_{1,b,c}$, $\mathbb S_{a,1,c}$ and $\mathbb S_{a,b,1}$. The dual of this convex hull is
the intersection of dual cones. Therefore, we have the the following. Recall that $\varrho$ is genuinely entangled
if it is not bi-separable.

\begin{corollary}\label{gen-wit}
A state $\varrho\in M_A\ot M_B\ot M_C$ is genuinely entangled if and only if there exists
$\phi\in\mathbb P_{1,b,c}\cap\mathbb P_{a,1,c}\cap\mathbb P_{a,b,1}$
satisfying $\lan \varrho,\phi\ran <0$.
\end{corollary}

Therefore, an $abc\times abc$ self-adjoint matrix $W$ is a witness for genuine entanglement if and only if
it is the Choi matrix of a bi-linear map which belongs to $\mathbb P_{1,b,c}\cap\mathbb P_{a,1,c}\cap\mathbb P_{a,b,1}$.
In the next section, we construct such examples for three qubit case of $a=b=c=2$.

On the other hand, a tri-partite state $\varrho$
is said to be fully bi-separable if it is bi-separable with respect to any possible bi-partitions.
This is the case if and only if it is in the intersection of $\mathbb S_{1,b,c}$, $\mathbb S_{a,1,c}$ and $\mathbb S_{a,b,1}$.
See \cite{kye_3_qubit,brunner} for examples of three qubit fully bi-separable states which are not genuinely separable.

Recall that the notion of $s$-positivity of a linear map is invariant under taking the dual map.
We proceed to consider what happens for $(p,q,r)$-positivity of bi-linear maps.
For a permutation $\sigma$ in the set $\{A,B,C\}$, we define the bi-linear map
$$
\phi^\sigma: M_{\sigma A}\times M_{\sigma B}\to M_{\sigma C}
$$
by
$$
\lan \phi^\sigma (x_{\sigma A}, x_{\sigma B}), x_{\sigma C}\ran
=\lan \phi(x_A,x_B),x_C\ran,\qquad
x_A\in M_A,\ x_B\in M_B,\ x_C\in M_C.
$$
The isomorphism from $M_A\ot M_B\ot M_C$ onto $M_{\sigma A}\ot M_{\sigma B}\ot M_{\sigma C}$
given by the flip map under $\sigma$
will be denote by $\varrho\mapsto \varrho^\sigma$.
By Corollary \ref{permutation}, we have
$$
\SN (\varrho) \le (p,q,r)\ \Longleftrightarrow\
\SN (\varrho^\sigma) \le (p,q,r)^\sigma.$$
 Since
$\langle \phi, \varrho \rangle = \langle \phi^\sigma, \varrho^\sigma \rangle$,
the following is immediate by Theorem \ref{dual}.

\begin{corollary}\label{dual-per}
Suppose that $\varrho\in M_A\ot M_B\ot M_C$ and $\phi:M_A\times M_B\to M_C$ is a bi-linear map.
Then, $\phi$ is $(p,q,r)$-positive if and only if
$\phi^\sigma$ is $(p,q,r)^\sigma$-positive.
\end{corollary}

By Proposition \ref{restriction} and Corollary \ref{equal},
we also have the following:
\begin{itemize}
\item
$\SN(\varrho) \le (p,q,r)$ and $r \ge pq$, then $\SN(\varrho) \le (p,q,pq)$.
\item
If $\SN(\varrho) \le (1,q,r)$, then $\SN(\varrho) \le (1,q \wedge r, q \wedge r)$.
\end{itemize}
This reflects Proposition \ref{redundant1} and Proposition \ref{redundant2}, respectively, by duality.
Applying Corollary \ref{dual-per}, it is clear that the role of $(p,q,r)$  may be permuted together with $(A,B,C)$
in Propositions \ref{redundant1} and \ref{redundant2}
if domains and ranges are matrix algebras.

\begin{corollary}
Suppose that $\phi:M_A\times M_B\to M_C$ is a bi-linear map. For $q,r \in \mathbb N$, the following are equivalent:
\begin{enumerate}
\item[(i)] $\phi$ is $(p,q,r)$-positive for all $p \in \mathbb N$;
\item[(ii)] $\phi$ is $(p,q,r)$-positive for some $p \ge qr$;
\item[(iii)] $\phi$ is $(qr,q,r)$-positive.
\end{enumerate}
Moreover, similar equivalence relations hold for fixed $p,r \in \mathbb N$.
\end{corollary}

\begin{corollary}
Suppose that $\phi:M_A\times M_B\to M_C$ is a bi-linear map. Then,
$\phi$ is $(p,q,1)$-positive if and only if $\phi$ is $(p \wedge q,p \wedge q, 1)$-positive.
\end{corollary}

The following tells us that
the inclusion relations between cones $\mathbb P_{p,q,r}$ are given by product order of triplets $(p,q,r)$,
and they are distinct for different triplets in $\Sigma_{a,b,c}$.

\begin{proposition}\label{order}
For $(p_1,q_1,r_1)$ and $(p_2,q_2,r_2)$ in $\Sigma_{a,b,c}$, the following are equivalent:
\begin{enumerate}
\item[(i)]
$\mathbb S_{p_1,q_1,r_1}\subset \mathbb S_{p_2,q_2,r_2}$;
\item[(ii)]
$\mathbb P_{p_1,q_1,r_1}\supset \mathbb P_{p_2,q_2,r_2}$;
\item[(iii)]
$p_1\le p_2$, $q_1\le q_2$ and $r_1\le r_2$.
\end{enumerate}
In particular, we have
$\mathbb S_{p_1,q_1,r_1}= \mathbb S_{p_2,q_2,r_2}$ if and only if
$\mathbb P_{p_1,q_1,r_1}= \mathbb P_{p_2,q_2,r_2}$ if and only if
$(p_1,q_1,r_1)=(p_2,q_2,r_2)$.
\end{proposition}

\begin{proof}
The equivalence (i) $\Longleftrightarrow$ (ii) and the implication (iii) $\Longrightarrow$ (i) follow from Theorem \ref{dual}
and Theorem \ref{SR}, respectively.

For any  $(p_1,q_1,r_1)\in \Sigma_{a,b,c}$,
we have constructed in Proposition \ref{restriction} a vector
$\xi \in \cl H_A \otimes \cl H_B \otimes \cl H_C$ with $\SR (\xi)=(p_1, q_1, r_1)$, and so
$|\xi \rangle \langle \xi | \in \mathbb S_{p_1,q_1,r_1}$. If
$|\xi \rangle \langle \xi | = \sum_k |\xi_k \rangle \langle \xi_k |$
then each $|\xi_k\ran$ is a scalar multiplication of $|\xi\ran$.
Therefore, if $p_2<p_1$ then $|\xi \rangle \langle \xi|$ never belongs to
$S_{p_2,q_2,r_2}$, and so
$\mathbb S_{p_1,q_1,r_1}\subset \mathbb S_{p_2,q_2,r_2}$ does not hold. The same is true when $q_2<q_1$ or $r_2<e_1$.
This completes the proof of the implication (i) $\Longrightarrow$ (iii).
\end{proof}

In order to construct bi-linear maps, we need to identify them by various objects
like functionals, matrices and linear maps.
For this purpose, we consider the following algebraic isomorphisms:
$$
\xymatrix{\mathcal{BL}
(M_A, M_B; M_C) \ar[r]^{\simeq}
& \mathcal L (M_A \otimes M_B, M_C) \ar[r]^-{\simeq} \ar[d]^{\simeq}
& (M_{ABC})^*\ar[r]^{\simeq} \ar[d]^{\simeq}
& M_{ABC} \\
& \mathcal L (M_{AB}, M_C) &  \mathcal L (M_A, M_{BC}) & }
$$
We denote by
$$
\begin{gathered}
 A_\phi \in \mathcal L (M_A \otimes M_B, M_C),\qquad B_\phi \in (M_{ABC})^*, \qquad C_\phi \in M_{ABC},\\
 D_\phi \in \mathcal L (M_{AB}, M_C), \qquad E_\phi \in  \mathcal L (M_A, M_{BC})
 \end{gathered}
$$
elements associated with a bilinear map
$\phi : M_A \otimes M_B \to M_C$
by the above isomorphisms, respectively.
It is worth to note that $C_\phi$ is the Choi matrix of $\phi$.
The properties of $A_\phi$, $B_\phi$ and $C_\phi$ corresponding to $(p,q,r)$-positivity of $\phi$ are
characterized in Theorem \ref{linearization}, Theorem \ref{dual} and Corollary \ref{ChoiBlock}, respectively.

\begin{theorem}\label{var-iso}
Suppose that $\phi : M_A \times M_B \to M_C$ is a bi-linear map. The following are equivalent:
\begin{enumerate}
\item[(i)] $\phi$ is $(p,q,r)$-positive;
\item[(ii)] $(D_\phi)_r: M_r(M_{AB}) \to M_r(M_C)$ maps $\mathbb S_{r,p,q}$ into $M_r(M_C)^+$;
\item[(iii)] $(E_\phi)_p : M_p(M_A) \to M_p(M_{BC})$ maps positive matrices to $(p,q,r)$-block positive matrices.
\end{enumerate}
\end{theorem}

\begin{proof}
(i) $\Longleftrightarrow$ (ii). By Theorem \ref{linearization}, $\phi : M_A \times M_B \to M_C$ is $(p,q,r)$-positive
if and only if the map
$$
(A_\phi)_r : M_r({\rm OMAX}^p(M_A) \otimes_{\max} {\rm OMAX}^q(M_B)) \to M_r(M_C)
$$
is positive.
On the other hand, we have
$$
\begin{aligned} \mathbb S_{r,p,q} & = ({\rm OMAX}^r(M_r) \otimes_{\max} {\rm OMAX}^p(M_A) \otimes_{\max} {\rm OMAX}^q(M_B))^+ \\
&=  M_r({\rm OMAX}^p(M_A) \otimes_{\max} {\rm OMAX}^q(M_B))^+.
\end{aligned}
$$
by Theorem \ref{SN}.

(i) $\Longleftrightarrow$ (iii). By (\ref{duality-max}), we have a complete order isomorphism
$$\begin{aligned}
 (\OMAX^p(M_A) \otimes_{\max}& OMAX^q (M_B) \otimes_{\max} \OMAX^r(M_C)))^* \\
\simeq & \mathcal L (\OMAX^p(M_A), (\OMAX^q(M_B) \otimes_{\max} \OMAX^r (M_C))^*).
\end{aligned}$$
Thus, $\phi : M_A \times M_B \to M_C$ is $(p,q,r)$-positive if and only if $B_\phi$ is positive in the left side
if and only if $E_\phi$ is completely positive in the right side
if and only if
$$
(E_\phi)_p : M_p(M_A) \to M_p((\OMAX^q(M_B) \otimes_{\max} \OMAX^r (M_C))^*)
$$
is positive. By \cite[Theorem 3.7]{X}, we have $M_p = \OMIN^p (M_p)$.
Therefore, we have the complete order isomorphism
$$
\begin{aligned}
M_p((\OMAX^q(M_B) \otimes_{\max} &\OMAX^r (M_C))^*)\\
& \simeq
{\rm OMIN}^p(M_p) \otimes_{\min} \OMIN^q(M_B) \otimes_{\min} \OMIN^r (M_C)
\end{aligned}
$$
by (\ref{duality-super}) and \cite[Proposition 1.16]{FP}.
The conclusion follows from Corollary \ref{ChoiBlock}.
\end{proof}

\begin{corollary}\label{special}
For a bi-linear map $\phi:M_A\times M_B\to M_C$, we have the following:
\begin{enumerate}
\item[(i)]
$\phi$ is $(a,b,r)$-positive if and only if $D_\phi$ is $r$-positive.
\item[(ii)]
$\phi$ is $(p,q,1)$-positive if and only if
$D_\phi$ maps $\mathbb V_{p \wedge q}$ into $M_C^+$.
\item[(iii)]
$\phi$ is $(p,b,c)$-positive if and only if
$E_\phi$ is $p$-positive.
\item[(iv)]
$\phi$ is $(1,q,r)$-positive if and only if
$E_\phi$ maps $M_A^+$ to $q \wedge r$-block positive matrices.
\end{enumerate}
\end{corollary}

If we express $(p,q,r)$-positive bi-linear maps by their Choi matrices, then they are witnesses
for entanglement $\varrho$ which does not satisfy $\SN(\varrho)\le (p,q,r)$. Therefore, it is very useful
to know how the Choi matrix is changed when we take the dual $\phi^\sigma$ with respect to the permutation $\sigma$.
For this purpose, we see easily that the following diagram commutes:
\begin{equation}\label{CD}
\begin{CD}
{\mathcal H}_A\ot {\mathcal H}_B\ot{\mathcal H}_C @>C_\phi>> {\mathcal H}_A\ot {\mathcal H}_B\ot{\mathcal H}_C\\
@VV V @VV V\\
{\mathcal H}_{\sigma A}\ot {\mathcal H}_{\sigma B}\ot{\mathcal H}_{\sigma C} @>C_{\phi^\sigma}>>
{\mathcal H}_{\sigma A}\ot {\mathcal H}_{\sigma B}\ot{\mathcal H}_{\sigma C}
\end{CD}
\end{equation}
where the vertical arrows are the flip operator under the permutation $\sigma$.
Therefore, taking dual with respect to a permutation
corresponds to changing the order of columns and rows in the Choi matrices.

%%%%%%%%%%%%%%%%%%%%%%%%%%%%%%%%%%%%%%%%%%%%%%%%%%%%%%%%%%%%%%%%%%%%%%%%%%%%%%%%%%%%%%%%%%%%%%%%%%%%%%%%%%%%%%%%%%
%%%%%%%%%%%%%%%%%%%%%%%%%%%%%%%%%%%%%%%%%%%%%%%%%%%%%%%%%%%%%%%%%%%%%%%%%%%%%%%%%%%%%%%%%%%%%%%%%%%%%%%%%%%%%%%%%%
%%%%%%%%%%%%%%%%%%%%%%%%%%%%%%%%%%%%%%%%%%%%%%%%%%%%%%%%%%%%%%%%%%%%%%%%%%%%%%%%%%%%%%%%%%%%%%%%%%%%%%%%%%%%%%%%%%
%%%%%%%%%%%%%%%%%%%%%%%%%%%%%%%%%%%%%%%%%%%%%%%%%%%%%%%%%%%%%%%%%%%%%%%%%%%%%%%%%%%%%%%%%%%%%%%%%%%%%%%%%%%%%%%%%%
%%%%%%%%%%%%%%%%%%%%%%%%%%%%%%%%%%%%%%%%%%%%%%%%%%%%%%%%%%%%%%%%%%%%%%%%%%%%%%%%%%%%%%%%%%%%%%%%%%%%%%%%%%%%%%%%%%
\section{Examples}

In this section, we exhibit examples of $(p,q,r)$-positive bi-linear maps between
$2\times 2$ matrices.
When we write down Choi matrices, we always use the lexicographic order:
\begin{equation}\label{qubit-basis}
|000\ran,\quad
|001\ran,\quad
|010\ran,\quad
|011\ran,\quad
|100\ran,\quad
|101\ran,\quad
|110\ran,\quad
|111\ran,
\end{equation}
for a basis of $\HA\ot\HB\ot\HC= \mathbb C^8$. Motivated by examples in \cite{kye_3_qubit}, we consider the following $8\times 8$ matrix
\begin{equation}\label{exam}
\left(\begin{matrix}
s_1 &\cdot &\cdot &\cdot &\cdot &\cdot &\cdot &u_1\\
\cdot &s_2 &\cdot &\cdot &\cdot &\cdot &u_2 &\cdot \\
\cdot &\cdot &s_3 &\cdot &\cdot &u_3 &\cdot &\cdot \\
\cdot &\cdot &\cdot &s_4 &u_4 &\cdot &\cdot &\cdot \\
\cdot &\cdot &\cdot &\bar u_4 &t_4 &\cdot &\cdot &\cdot \\
\cdot &\cdot &\bar u_3 &\cdot &\cdot &t_3 &\cdot &\cdot \\
\cdot &\bar u_2 &\cdot &\cdot &\cdot &\cdot &t_2 &\cdot \\
\bar u_1 &\cdot &\cdot &\cdot &\cdot &\cdot &\cdot &t_1
\end{matrix}\right),
\end{equation}
with nonnegative numbers $s_i, t_i$ and complex numbers $u_i$, for $i=1,2,3,4$,
where $\cdot$ denotes zero.
This is the Choi matrix of the bi-linear map $\phi:M_2\times M_2\to M_2$ which sends
the pair $([x_{ij}], [y_{k\ell}])\in M_2\times M_2$ to
$$
\left(\begin{matrix}
s_1x_{11}y_{11}+s_3x_{11}y_{22}+t_4x_{22}y_{11}+t_2x_{22}y_{22}
& u_1x_{12}y_{12} +u_3x_{12}y_{21}+\bar u_4x_{21}y_{12} +\bar u_2 x_{21}y_{21}\\
u_2x_{12}y_{12} +u_4x_{12}y_{21}+\bar u_3x_{21}y_{12} +\bar u_1 x_{21}y_{21}
&s_2x_{11}y_{11}+s_4x_{11}y_{22}+t_3x_{22}y_{11}+t_1x_{22}y_{22}
\end{matrix}\right).
$$
By Proposition \ref{restriction} and duality, all possible kinds of positivity of $\phi$ are given by
$$
(2,2,2),\qquad (1,2,2),\qquad (2,1,2),\qquad (2,2,1),\qquad (1,1,1).
$$
By Theorem \ref{Choi-iso}, we see that $\phi$ is $(2,2,2)$-positive if and only if its Choi matrix is positive
if and only if
$$
\sqrt{s_it_i}\ge |u_i|
$$
for each $i=1,2,3,4$.

\begin{lemma}\label{ex-lemma}
Let $a, c, b, d \ge 0$ and $\omega , z \in \mathbb C$. The inequality
$$
\sqrt{ab}+\sqrt{cd} \ge |\omega |+|z|,
$$
holds if and only if the matrix
$$
\begin{pmatrix}
a+d|\alpha|^2 &\omega \bar\alpha +\bar z\alpha\\
\bar \omega \alpha + z\bar\alpha & c+b|\alpha|^2
\end{pmatrix}
$$
is positive for all $\alpha \in \mathbb C$.
\end{lemma}

\begin{proof}
($\Longrightarrow$) By the Cauchy-Schwartz inequality, we have
$$
\begin{aligned}
| \omega \bar\alpha +\bar z\alpha |
&\le |\omega ||\bar\alpha |+|\bar z||\alpha|\\
&\le (\sqrt{ab}+\sqrt{cd})|\alpha|\\
&=\sqrt{a}(\sqrt{b}|\alpha|) +(\sqrt{d}|\alpha|)\sqrt{c}\\
&\le\sqrt{a+d|\alpha|^2}\sqrt{c+b|\alpha|^2}.
\end{aligned}
$$

($\Longleftarrow$) We first consider the case when $bd\neq 0$.
Let $$\omega z=|\omega z|e^{i\theta} \qquad \text{and} \qquad {\alpha_0:= \left(\frac{ac}{bd}\right)^{\frac 14}e^{i\theta/2}}.$$
Then we have
$$
\begin{aligned}
\sqrt{a+d|\alpha_0|^2}\sqrt{c+b|\alpha_0|^2}
&=\sqrt{a+ \dfrac{d\sqrt{ac}}{\sqrt{bd}}}
  \sqrt{c+\dfrac{b\sqrt{ac}}{\sqrt{bd}}}  \\
&=\left( \sqrt{\frac{a}{b}}\, (\sqrt{ab}+\sqrt{cd})
   \sqrt{\frac{c}{d}}\, (\sqrt{ab}+\sqrt{cd}) \right)^{\frac 12}\\
&=|\alpha_0|(\sqrt{ab}+\sqrt{cd}).
\end{aligned}
$$
On the other hand, we have
$$
\begin{aligned}
|\omega \bar\alpha_0+\bar z\alpha_0|
&=\sqrt{|\omega \bar\alpha_0+\bar z\alpha_0|^2} \\
&=\left( |\omega |^2|\alpha_0|^2+\omega z\bar\alpha_0^2+\bar \omega \bar z\alpha_0^2 +|z|^2|\alpha_0|^2\right)^{\frac 12}\\
&=\left( |\omega |^2|\alpha_0|^2+ 2|\omega z||\alpha_0|^2 +|z|^2|\alpha_0|^2\right)^{\frac 12}
=|\alpha_0|(|\omega |+|z|).
\end{aligned}
$$
Hence, the inequality
$$
\sqrt{ab}+\sqrt{cd} \ge |\omega |+|z|,
$$
holds if $b d \ne 0$.

When $b d=0$, we take $\varepsilon>0$ and apply the above to the positive matrix
$$
\begin{pmatrix}
a+(d+\varepsilon) |\alpha|^2 &\omega \bar\alpha +\bar z\alpha\\
\bar \omega \alpha + z\bar\alpha & c+(b+\varepsilon)|\alpha|^2
\end{pmatrix},
$$
which yields
$$
\sqrt{a(b+\varepsilon)}+\sqrt{c(d+\varepsilon)} \ge |\omega |+|z|
$$
for any $\varepsilon>0$.
\end{proof}

\begin{theorem}\label{ex-th}
Suppose that $\phi:M_2\times M_2\to M_2$ is a bilinear map given by its Choi matrix as {\rm (\ref{exam})}, and
consider the inequalities
\begin{equation}\label{ex-ineq}
\sqrt{s_it_i}+\sqrt{s_jt_j} \ge |u_i|+|u_j|.
\end{equation}
Then we have the following:
\begin{enumerate}
\item[(i)]
$\phi$ is $(1,2,2)$-positive if and only if {\rm (\ref{ex-ineq})} hold for $(i,j)=(1,4)$ and $(2,3)$.
\item[(ii)]
$\phi$ is $(2,1,2)$-positive if and only if {\rm (\ref{ex-ineq})} hold for $(i,j)=(1,3)$ and $(2,4)$.
\item[(iii)]
$\phi$ is $(2,2,1)$-positive if and only if {\rm (\ref{ex-ineq})} hold for $(i,j)=(1,2)$ and $(3,4)$.
\end{enumerate}
\end{theorem}

\begin{proof}
We will consider the linear map $E_\phi:M_2\to M_4$ in Theorem \ref{var-iso}. The map $E_\phi$ sends
a rank one positive matrix
$P_\alpha=\left(\begin{matrix}1&\bar\alpha\\ \alpha &|\alpha|^2\end{matrix}\right)$ to
\begin{equation}\label{ex-4x4}
\left(\begin{matrix}
s_1+t_4|\alpha|^2 &\cdot &\cdot &u_1\bar\alpha+\bar u_4\alpha\\
\cdot &s_2+t_3|\alpha|^2 &u_2\bar\alpha+\bar u_3\alpha &\cdot\\
\cdot &u_3\bar\alpha +\bar u_2\alpha &s_3+t_2|\alpha|^2 &\cdot\\
u_4\bar\alpha +\bar u_1\alpha &\cdot &\cdot  &s_4+t_1|\alpha|^2
\end{matrix}\right).
\end{equation}
We see that $\phi$ is $(1,2,2)$-positive if and only if $E_\phi$ is positive by Corollary \ref{special} (iii),
if and only if the above matrix is positive for each complex number $\alpha$. From this, we get the statement (i)
by Lemma \ref{ex-lemma}.

For the statement (ii), we consider the permutation $\sigma$ which sends $(A,B,C)$ to $(B,C,A)$. Then $\phi$
is $(2,1,2)$-positive if and only if $\phi^\sigma$ is $(1,2,2)$-positive. The basis (\ref{qubit-basis}) is
changed to
$$
|000\ran,\quad
|010\ran,\quad
|100\ran,\quad
|110\ran,\quad
|001\ran,\quad
|011\ran,\quad
|101\ran,\quad
|111\ran,
$$
by the flip operation under $\sigma$. Therefore, we have $C_{\phi^\sigma}=UC_\phi U^*$ with the
$8\times 8$ unitary $U$ which sends a standard ordered basis to an ordered basis
$\{e_1,e_3,e_5,e_7,e_2,e_4,e_6,e_8\}$. Therefore, we have
\begin{equation}\label{yyyy}
C_{\phi^\sigma}=
\left(\begin{matrix}
s_1 &\cdot &\cdot &\cdot &\cdot &\cdot &\cdot &u_1\\
\cdot &t_4 &\cdot &\cdot &\cdot &\cdot &\bar u_4 &\cdot \\
\cdot &\cdot &s_2 &\cdot &\cdot & u_2 &\cdot &\cdot \\
\cdot &\cdot &\cdot &t_3 &\bar u_3 &\cdot &\cdot &\cdot \\
\cdot &\cdot &\cdot &u_3 &s_3 &\cdot &\cdot &\cdot \\
\cdot &\cdot &\bar u_2 &\cdot &\cdot &t_2 &\cdot &\cdot \\
\cdot & u_4 &\cdot &\cdot &\cdot &\cdot &s_4 &\cdot \\
\bar u_1 &\cdot &\cdot &\cdot &\cdot &\cdot &\cdot &t_1
\end{matrix}\right),
\end{equation}
and the result follows from (i).

If we consider the permutation $\sigma$ which send $(A,B,C)$ to $(C,A,B)$ then the standard ordered basis
goes to an ordered basis $\{e_1,e_5, e_2,e_6,e_3,e_7,e_4,e_8\}$, and we have
\begin{equation}\label{xxxx}
C_{\phi^\sigma}=
\left(\begin{matrix}
s_1 &\cdot &\cdot &\cdot &\cdot &\cdot &\cdot &u_1\\
\cdot &s_3 &\cdot &\cdot &\cdot &\cdot &u_3 &\cdot \\
\cdot &\cdot &t_4 &\cdot &\cdot &\bar u_4 &\cdot &\cdot \\
\cdot &\cdot &\cdot &t_2 &\bar u_2 &\cdot &\cdot &\cdot \\
\cdot &\cdot &\cdot &u_2 &s_2 &\cdot &\cdot &\cdot \\
\cdot &\cdot &u_4 &\cdot &\cdot &s_4 &\cdot &\cdot \\
\cdot &\bar u_3 &\cdot &\cdot &\cdot &\cdot &t_3 &\cdot \\
\bar u_1 &\cdot &\cdot &\cdot &\cdot &\cdot &\cdot &t_1
\end{matrix}\right).
\end{equation}
Therefore, (iii) comes from (i) again.
\end{proof}

As for the $(1,1,1)$-positivity, we use $E_\phi$ in Proposition \ref{var-iso}
together with Lemma \ref{ex-lemma}, to get the following:

\begin{theorem}\label{111}
Suppose that $\phi:M_2\times M_2\to M_2$ is a bilinear map given by its Choi matrix as {\rm (\ref{exam})}.
Then, $\phi$ is $(1,1,1)$-positive if and only if the inequality
$$
\begin{aligned}
\sqrt{(s_1+t_4|\alpha|^2)(s_4+t_1|\alpha|^2)}
+&\sqrt{(s_2+t_3|\alpha|^2)(s_3+t_2|\alpha|^2)}\\
&\phantom{XXXX}\ge |u_1\bar\alpha +\bar u_4 \alpha|+|u_2 \bar \alpha + \bar u_3 \alpha|
\end{aligned}
$$ holds
for each $\alpha\in\mathbb C$. In particular, this holds when
$$
\sum_{i=1}^4 \sqrt{s_it_i}\ge \sum_{i=1}^4|u_i|.
$$
\end{theorem}

\begin{proof}
The linear map given with the Choi matrix (\ref{ex-4x4}) sends
a rank one positive matrix
$P_\beta=\left(\begin{matrix}1&\bar\beta\\ \beta &|\beta|^2\end{matrix}\right)$ to
$$
\begin{pmatrix}
(s_1+t_4|\alpha|^2) +(s_3+t_2|\alpha|^2)|\beta|^2 & (u_1\bar\alpha +\bar u_4 \alpha) \bar \beta + (u_3\bar\alpha +\bar u_2 \alpha) \beta \\
(u_4 \bar \alpha + \bar u_1 \alpha) \beta + (u_2 \bar \alpha + \bar u_3 \alpha) \bar \beta & (s_2+t_3|\alpha|^2) +(s_4+t_1|\alpha|^2)|\beta|^2
\end{pmatrix}.
$$
The bilinear map $\phi$ is $(1,1,1)$-positive if and only if the matrix (\ref{ex-4x4}) is block positive for all
$\alpha \in \mathbb C$ by Corollary \ref{special} (iv), if and only if the above matrix is positive for all $\alpha, \beta \in \mathbb C$.
It is equivalent to the inequality by Lemma \ref{ex-lemma}.

For the last statement, we note
$$
\begin{aligned}
|u_1\bar\alpha +\bar u_4 \alpha|&+|u_2 \bar \alpha + \bar u_3 \alpha|
\le( \sum_{i=1}^4|u_i|) |\alpha|
\le (\sum_{i=1}^4 \sqrt{s_it_i})|\alpha|\\
&=\sqrt{s_1}(\sqrt{t_1}|\alpha|) +(\sqrt{t_4}|\alpha|)\sqrt{s_4} +\sqrt{s_2}(\sqrt{t_2}|\alpha|) +(\sqrt{t_3}|\alpha|)\sqrt{s_3}\\
&\le \sqrt{(s_1+t_4|\alpha|^2)(s_4+t_1|\alpha|^2)}
+\sqrt{(s_2+t_3|\alpha|^2)(s_3+t_2|\alpha|^2)}
\end{aligned}
$$
by the Cauchy-Schwartz inequality.
\end{proof}

The converse of the last statement does not hold, as the examples in
\cite{kye_3_qubit} show.
Now, we have bunch of examples which distinguish various notions of positivity and their intersections.
For example, if we take
$$
\sqrt{s_1t_1}=0,\quad
\sqrt{s_2t_2}=1,\quad
\sqrt{s_3t_3}=1,\quad
\sqrt{s_4t_4}=2,\quad
|u_i|=1,\ i=1,2,3,4,
$$
then the resulting map is $(1,2,2)$-positive but neither $(2,1,2)$ nor $(2,2,1)$-positive.
If we take
$$
\sqrt{s_1t_1}=0,\quad
\sqrt{s_2t_2}=0,\quad
\sqrt{s_3t_3}=2,\quad
\sqrt{s_4t_4}=2,\quad
|u_i|=1,\ i=1,2,3,4,
$$
then the map is both $(1,2,2)$ and $(2,1,2)$-positive but not $(2,2,1)$-positive. If we put
$\sqrt{s_it_i}=0$ for $i=1,2,3$ and $\sqrt{s_4 t_4}=4$, then we may find an example of $(1,1,1)$-positive map which is not $(p,q,r)$-positive
for other $(p,q,r)$. Finally, we may also get an example of $(p,q,r)$-positive map for each $(p,q,r)$ in $\Sigma_{2,2,2}$
except for $(2,2,2)$  by putting $\sqrt{s_1t_1}=0$ and $\sqrt{s_it_i}=2$ for $i=2,3,4$.

By Theorem \ref{ex-th} and Corollary \ref{gen-wit}, we have the following.

\begin{corollary}\label{bi-sep-wnt}
Suppose that $\phi:M_2\times M_2\to M_2$ is a bilinear map given by its Choi matrix as {\rm (\ref{exam})}. Then
the following are equivalent:
\begin{enumerate}
\item[(i)]
$\lan\varrho,\phi\ran\ge 0$ for each bi-separable three qubit state $\varrho$;
\item[(ii)]
inequality {\rm (\ref{ex-ineq})} holds for each possible choice of $i,j$ with $i\neq j$ from $\{1,2,3,4\}$.
\end{enumerate}
\end{corollary}

Therefore, if $W$ is the Choi matrix of a bi-linear map $\phi$ satisfying the conditions in Corollary
\ref{bi-sep-wnt} and $\lan\varrho,\phi\ran <0$ then $\varrho$ is a genuinely entangled state.
In this sense, those $W$ are genuine entanglement witnesses.
We consider the following matrix
$$
W=\left(\begin{matrix}
\cdot &\cdot &\cdot &\cdot &\cdot &\cdot &\cdot &-1\\
\cdot &s &\cdot &\cdot &\cdot &\cdot &\cdot &\cdot \\
\cdot &\cdot &s &\cdot &\cdot &\cdot &\cdot &\cdot \\
\cdot &\cdot &\cdot &s &\cdot &\cdot &\cdot &\cdot \\
\cdot &\cdot &\cdot &\cdot &t &\cdot &\cdot &\cdot \\
\cdot &\cdot &\cdot &\cdot &\cdot &t &\cdot &\cdot \\
\cdot &\cdot &\cdot &\cdot &\cdot &\cdot &t &\cdot \\
-1 &\cdot &\cdot &\cdot &\cdot &\cdot &\cdot &\cdot
\end{matrix}\right),
$$
with $st=1$. We note that this $W$ is a genuine entanglement
witness, since it satisfies the condition (ii) of Corollary
\ref{bi-sep-wnt}. Consider the GHZ type pure state \cite{abls}
given by the vector
$$
|\psi_{\text{\rm GHZ}}\ran
= \lambda_0|000\ran
+\lambda_1 e^{i\theta} |100\ran
+\lambda_2 |101\ran
+\lambda_3 |110\ran
+\lambda_4 |111\ran,
$$
with $\lambda_i\ge 0$. Then we see that
$$
\lan | \psi_{\text{\rm GHZ}}\ran\lan \psi_{\text{\rm GHZ}}|,W\ran=
t(\lambda_1^2+\lambda_2^2+\lambda_3^2) -2\lambda_0\lambda_4.
$$
Taking arbitrary small $t>0$, we see that $W$ detects every  $| \psi_{\text{\rm GHZ}}\ran\lan \psi_{\text{\rm GHZ}}|$ with $\lambda_0\lambda_4>0$.
These include GHZ type entangled pure states in the classification \cite{abls} of three qubit states.

%%%%%%%%%%%%%%%%%%%%%%%%%%%%%%%%%%%%%%%%%%%%%%%%%%%%%%%%%%%%%%%%%%%%%%%%%%%%%%%%%%%%%%%%%%%%%%%%%%%%%%%%%%%%%%%%%%
%%%%%%%%%%%%%%%%%%%%%%%%%%%%%%%%%%%%%%%%%%%%%%%%%%%%%%%%%%%%%%%%%%%%%%%%%%%%%%%%%%%%%%%%%%%%%%%%%%%%%%%%%%%%%%%%%%
%%%%%%%%%%%%%%%%%%%%%%%%%%%%%%%%%%%%%%%%%%%%%%%%%%%%%%%%%%%%%%%%%%%%%%%%%%%%%%%%%%%%%%%%%%%%%%%%%%%%%%%%%%%%%%%%%%
%%%%%%%%%%%%%%%%%%%%%%%%%%%%%%%%%%%%%%%%%%%%%%%%%%%%%%%%%%%%%%%%%%%%%%%%%%%%%%%%%%%%%%%%%%%%%%%%%%%%%%%%%%%%%%%%%%
%%%%%%%%%%%%%%%%%%%%%%%%%%%%%%%%%%%%%%%%%%%%%%%%%%%%%%%%%%%%%%%%%%%%%%%%%%%%%%%%%%%%%%%%%%%%%%%%%%%%%%%%%%%%%%%%%%
\section{Discussion}

We have defined in Section 4 the notion of Schmidt rank for the tensor product of arbitrary number of vector spaces.
We also note that Theorem \ref{dual} is easily extended to $(n-1)$-linear maps and $n$-partite states in an obvious way,
with this definition.
This was also useful to clarify how the usual Schmidt rank for tensor of two spaces can be explained in our definition,
and explain bi-separability in $3$-partite cases.
In the $4$-partite case, it is possible with our definition to explain bi-separability according the bi-partition $4=1+3$ like $A$-$BCD$
bi-separability. But we cannot explain bi-separability according the bi-partition $4=2+2$.
It would be interesting to refine the notions of positivity and Schmidt numbers with which we may
explain all kinds of bi-separability for the cases of $n\ge 4$.

As for $(p,q,r)$-positive bi-linear maps, we have shown that different triplets $(p,q,r)$ in $\Sigma_{a,b,c}$ give
rise to different convex cones $\mathbb P_{p,q,r}$, and we give concrete examples in $2\times 2$ matrices.
It would be interesting to give concrete examples in higher dimensions.
Recall that that the first examples which distinguish $2$-positivity and $3$-positivity in the $3\times 3$ matrices
was given by Choi \cite{choi72}. See also \cite{cho-kye-lee}.

We would like to remind the readers that we did not define complete positivity for bi-linear maps.
It seems to be reasonable to say that a bi-linear map is completely positive when it is $(\infty,\infty,\infty)$-positive,
that is, $(p,q,r)$-positive for every triplet $(p,q,r)$, or equivalently, satisfies the condition (\ref{st-pos}) for each $p,q=1,2,\dots$.
We note that the term \lq complete positivity\rq\ for multi-linear maps already used in \cite{CS,heo} in totally
different contexts from ours. Furthermore, the authors of \cite{KPTT} call those bi-linear maps with the condition
(\ref{st-pos}) for each $p,q=1,2,\dots$ \lq jointly completely positive\rq\ maps. As for similar problems in terminologies
in the notions of complete boundedness for bi-linear maps, we refer to comments in \cite{PS}.

Anyway, we call temporarily $(\infty,\infty,\infty)$-positive bi-linear maps completely positive maps.
Then we may define various kinds of complete copositivity and decomposability for bi-linear maps
$\phi:M_A\times M_B\to M_C$ between matrix algebras. Recall that a linear map $\phi:M_A\to M_C$
is completely copositive if $\phi\circ\ttt_A$ is completely positive. This is the case if and only if
$\ttt_C\circ\phi$ is completely positive, where $\ttt_A$ and $\ttt_C$ denote the transpose maps in $M_A$ and $M_C$,
respectively. There are three kinds of complete copositivity according to the complete positivity of the maps
$$
\phi\circ (\idd_A\times\ttt_B),\qquad
\phi\circ (\ttt_A\times \idd_B),\qquad
\ttt_C\circ\phi
$$
for a bi-linear map $\phi:M_A\times M_B\to M_C$.

We say that a bi-linear map is decomposable if it is the sum of
a completely positive map and three kinds of completely copositive maps.
We would like to ask what kinds of $(p,q,r)$-positivity imply decomposability.
As for bi-linear maps in Theorem \ref{ex-th}, it is easy to see that they are decomposable
whenever they are $(1,2,2)$, $(2,1,2)$ or $(2,2,1)$-positive.
For examples of indecomposable $(1,1,1)$-positive bi-linear maps in $M_2$, we refer to \cite{kye_3_qubit}.
There is a long standing question which asks if every $2$-positive linear map
between $M_3$ is decomposable. See Corollary 4.3 in \cite{cho-kye-lee}. The dual question
asks if every $3\otimes 3$ PPT state has Schmidt number less than or equal to $2$.
This was conjectured in \cite{sbl}.
Finally, it would be interesting to define decomposability for general situations beyond matrix algebras,
as in the case of linear maps. See \cite{stormer}.

\end{document}